%% file: main.tex
\documentclass{ceurart}

\usepackage[T1]{fontenc}
\usepackage[utf8]{inputenc}
\usepackage[normalem]{ulem}
\usepackage[]{outlines}
\usepackage[dvipsnames]{xcolor}
\usepackage{mathpartir}
\usepackage{stmaryrd}
\usepackage{iftex}
\usepackage{amsmath,amssymb,amsthm,comment}

\PassOptionsToPackage{hyphens}{url}
\usepackage{hyperref}
\hypersetup{
  unicode,
  bookmarksnumbered,
  colorlinks,
  pdfborder={0 0 0},
  pdfmenubar=false,
  pdfpagemode=UseNone,
  pdfnonfullscreenpagemode=UseNone,
  bookmarksopen=false,
  breaklinks=true,
  linkcolor={green!50!black},
  citecolor={red!50!black},
  urlcolor={blue!50!black}
}

\usepackage{xspace}
\usepackage{nicefrac}
\usepackage{tikz}
\usetikzlibrary{arrows.meta}
\tikzset{>=Latex[]}
\usepackage{booktabs}
\usepackage{caption}
\usepackage{subcaption}
\usepackage[inline]{enumitem}
\usepackage[nolist,nohyperlinks]{acronym}
\usepackage{mathtools}
\usepackage{forest}

\input{macros}

\begin{document}

\copyrightyear{2026}
\copyrightclause{Copyright for this paper by its authors. Use permitted under
	Creative Commons License Attribution 4.0 International (CC BY 4.0).}

\conference{%
	24th International Workshop on Satisfiability Modulo Theories,
	July 24--25, 2026, Lisbon, Portugal
}

\title{%
	A Modern View on MCSat
}

\author[1]{Thomas Hader}[%
email=thomas.hader@tuwien.ac.at,
orcid=0009-0002-8920-5469
]
\cormark[1]
\author[1]{Theo Jauschneg}[
email=theo.jauschneg@tuwien.ac.at,
]
\author[1]{Daniela Kaufmann}[%
email=daniela.kaufmann@tuwien.ac.at,
orcid=0000-0002-5645-0292,
]
\author[1]{Laura Kov\'acs}[%
email=laura.kovacs@tuwien.ac.at,
orcid=0000-0002-8299-2714,
]

\address[1]{TU Wien,
	Favoritenstraße 9-11, 1040 Wien,
	Austria}

\cortext[1]{Corresponding author.}


\maketitle

\input{abstract}

\acresetall

\section{Introduction}\label{sec:intro}
\input{intro}

\section{Preliminaries}\label{sec:prelim}
\input{prelim}

\section{\acs{mcsat} Search Schema}\label{sec:rules}
\input{rules}

\section{\acs{mcsat} Theories}\label{sec:plugins}
\input{plugins}

\section{Conclusion}\label{sec:conclusion}
\input{conclusion}

\input{acks}

\begin{flushleft}
\footnotesize
\setlength{\bibsep}{1pt}
\bibliography{refs}
\end{flushleft}

\appendix

\section{Proofs}\label{sec:proofs}
\input{proofs}

\section{Examples}\label{sec:example}
\input{example}

\end{document}

%% file: macros.tex
\newacro{mcsat}[MCSat]{Model Constructing Satisfiability}
\newacro{dpll}[DPLL]{Davis–Putnam–Logemann–Loveland}
\newacro{cdcl}[CDCL]{Conflict-Driven Clause Learning}
\newacro{cnf}[CNF]{conjunctive normal form}
\newacro{dnf}[DNF]{disjunctive normal form}
\newacro{sat}[SAT]{Boolean satisfiability}
\newacro{smt}[SMT]{satisfiability modulo theory}
\newacro{cdcl-t}[CDCL(T)]{\acl{cdcl} with theory support}
\newacro{vsids}[VSIDS]{Variable State Independent Decaying Sum}
\newacro{evsids}[EVSIDS]{exponential \ac{vsids}}
\newacro{vmtf}[VMTF]{variable-move-to-front}
\newacro{cad}[CAD]{cylindrical algebraic decomposition}

\newcommand{\mcsat}{\ac{mcsat}\xspace}
\newcommand{\yices}{Yices2\xspace}

\newcommand{\smtrat}{SMT-RAT\xspace}
\newcommand{\zthree}{Z3\xspace}
\newcommand{\nlsat}{NLSAT\xspace}

\newcommand{\smt}{\ac{smt}\xspace}
\newcommand{\sat}{\ac{sat}\xspace}
\newcommand{\smtlib}{SMT-LIB\xspace}

\newcommand{\true}{\ensuremath{\mathsf{true}\xspace}}
\newcommand{\false}{\ensuremath{\mathsf{false}\xspace}}
\newcommand{\udef}{\ensuremath{\mathsf{undef}\xspace}}
\newcommand{\subterms}{\mathsf{subterms}\xspace}
\newcommand{\eval}{\ensuremath{\mathsf{eval}}\xspace}
\newcommand{\sort}{\ensuremath{\mathsf{sort}}\xspace}
\newcommand{\func}{\ensuremath{\mathsf{func}}\xspace}

\newcommand{\terms}{\ensuremath{\mathsf{terms}}\xspace}

\newtheorem{theorem}{Theorem}
\newtheorem{lemma}{Lemma}
\newtheorem{example}{Example}

\usepackage{tcolorbox}
\tcbuselibrary{skins, breakable}

\makeatletter
\newenvironment{rulebox}[1]{%
	\begin{tcolorbox}[
		title        = {%
			\def\@currentlabelname{#1}%
			\phantomsection\label{rule:#1}%
			\textbf{#1}%
		},
		colback      = white,
		colframe     = black!75,
		colbacktitle = black!85,
		coltitle     = white,
		fonttitle    = \bfseries,
		arc          = 4pt,
		boxrule      = 0.6pt,
		titlerule    = 0pt,
		breakable,
		after        = {\noindent},
		]%
	}{%
	\end{tcolorbox}%
}
\makeatother

\newtheorem*{runningexample}{Running Example}

\newcommand{\teq}{\approx\xspace}
\newcommand{\tneq}{\mathrel{\not\approx}\xspace}
\newcommand{\lorS}{\,\lor\,}
\newcommand{\landS}{\,\land\,}
\newcommand{\limpS}{\,\rightarrow\,}
\newcommand{\bvor}{\ensuremath{\mathsf{or}\xspace}}

\newcommand{\ssat}{\ensuremath{\mathsf{sat}}\xspace}
\newcommand{\susat}{\ensuremath{\mathsf{unsat}}\xspace}
\newcommand{\initstate}{\ensuremath{\state{\emptytrail}}\xspace}

%% file: abstract.tex
\begin{abstract}
The Model Constructing Satisfiability (MCSat) approach has shown strong performance in solving complex SMT problems, in particular in algebraic SMT theories such as non-linear integer and real arithmetic. 
In this paper we revisit the theory-independent MCSat framework as a proof system to provide a modern perspective that refines the original formulation of MCSat.

By closely formalizing the implementation of MCSat within the Yices2 SMT solver, we incorporate design decisions that diverge from those in the seminal MCSat paper and thereby capture the current state-of-the-art in MCSat-based SMT reasoning.
	
We present a general, theory-agnostic rule scheme for MCSat and instantiate it for several theories, including propositional logic, non-linear real arithmetic, and uninterpreted functions. 
We provide several detailed examples to illustrate the applicability of the presented calculus.
\end{abstract}

%% file: intro.tex
\acused{smt}

The \ac{mcsat} calculus represents an alternative approach to DPLL(T) in \smt solving by integrating model construction and theory reasoning within a single, unified framework.
It was initially developed under the name \nlsat as an efficient approach to solve the satisfiability problem for the theory of quantifier-free non-linear real arithmetic (NRA)~\cite{nlsat}.
When \nlsat proved to be particularly effective, it was extended to encompass arbitrary theories and later published as \mcsat~\cite{mcsat,mcsat2}.
A theoretical view on the \mcsat principles with focus on theory combination has been published in~\cite{cdsat}.

Unlike the traditional DPLL(T) approach, \ac{mcsat} maintains a combined candidate model in all present theories throughout the solving process and extends it incrementally, allowing for more implicit (automatic) handling of complex theory combinations~\cite{mcsat}.
The \ac{mcsat}-calculus uses conflict analysis to guide the search, drawing inspiration from \ac{cdcl} in \ac{sat} solving and conflict explanation techniques from constraint programming~\cite{Korovin2009}.
It, thus, can be seen as a generalization of the prevalent solving techniques from \sat solving to first-order theories where different theories are combined in a general search framework based on a common partial model. As the model search is handled by a common, theory-independent framework, individual theories are merely required to provide conflict explanation in case the partial assignment is detected inconsistent in the respective theory.

To this day \mcsat remains to be the state-of-the-art approach in solving NRA and has since been implemented in many \smt solvers, such as \yices~\cite{DBLP:conf/cav/Dutertre14} and \smtrat~\cite{smtrat} (as a stand-alone solver) as well as \zthree~\cite{DBLP:conf/tacas/MouraB08} (as a theory back-end for NRA in its original \nlsat version).
However, over years of development, current implementations of \mcsat may sufficiently differ from the formal presentation in the original literature~\cite{mcsat}. This poses additional challenges for interested readers in understanding state-of-the-art \mcsat procedures.
In this work, we revisit the \mcsat theory by presenting an updated definition of \mcsat as a proof system that closely captures the implementation of the \mcsat reasoning engine of the \yices \smt solver.
While an experienced reader will certainly see close similarity to the original formalization, we aim to provide a compact and example-driven presentation to enable novel readers an easy and understandable introduction to \mcsat in general as well as its \yices implementation in particular.
Based on \smtlib~\cite{smtlib}, we first provide necessary definitions (Section~\ref{sec:prelim}) to later define the search schema (Section~\ref{sec:rules}), followed by some examples for theories (Section~\ref{sec:plugins}).

\acused{smt}
\acused{mcsat}

%% file: prelim.tex
\newcommand{\dom}{\mathsf{dom}\xspace}

Along the definitions of \smtlib (version 2.6~\cite{smtlib}), we define the \smt problem as the check for satisfiability in a many-sorted first-order logic with equality, limited to the quantifier free fragment, augmented with \emph{background theories}.
A background theory specifies (one or many) \emph{sorts}.
Each sort $\mathcal{S}$ contains a \emph{domain} $D_\mathcal{S}$ which is a non-empty, possibly infinite, set.
We define $\dom(\mathcal{S}) = D_\mathcal{S}$. By slight abuse of notation, we sometimes simply write $\mathcal{S}$ for $D_\mathcal{S}$.
Logical expressions in an \smt problem are defined using the language of \emph{well-sorted terms}, thus each term is associated with a unique sort.
Terms are constructed using \emph{functions}. A function consists of a \emph{function symbol} (i.e. its name) and a \emph{signature}, an ordered list of sorts of parameters and the sort it returns. We write $f: \mathcal{S}_1\times\dots\times\mathcal{S}_n \rightarrow \mathcal{S}$ to denote a function with symbol $f$ with a signature that maps sorts $\mathcal{S}_1$ up to $\mathcal{S}_n$ to sort $\mathcal{S}$. A term $t$ generated by $f$ is of sort $\mathcal{S}$, denoted by $\sort(t) = \mathcal{S}$.
Functions without parameters are \emph{constants}.
An \emph{interpretation} $I$ is a total function that maps a function to its semantics, i.e. its meaning.
To apply the interpretation $I$ on an arbitrary formula or term $a$ we write $a^I$.
For a constant, this is an element of $\mathcal{S}$, for all other functions, this is a total function (in the mathematical sense).
Functions can either be \emph{interpreted} or \emph{uninterpreted}. Interpreted functions have a fixed interpretation defined by the respective theory, e.g.\ addition in the theory of reals. Finding the semantics of an uninterpreted function is part of the search problem.
A term $t$ defined by function $f$ is written as $f(t_1,\dots,t_n)$ where $t_1,\dots,t_n$ are terms, called \emph{sub-terms} of $t$.
We denote $\subterms(t) = \{t_1, \dots, t_n\}$.
We define the notion of \emph{interpreted} and \emph{uninterpreted} term accordingly.
We further denote by $\terms(t)$ the transitive closure of the sub-term relation.
For better readability, we denote (dis-)equality inside of a term with $\teq$ and $\tneq$.

An \smt \emph{formula} $\varphi$ is a term of the distinguished \emph{Boolean} sort $\mathbb{B}$.
Naturally, we have $\dom(\mathbb{B}) = \{\true, \false\}$ and we can define all regular Boolean operators as interpreted functions accordingly. We write $\bot$ and $\top$ to denote the constant terms evaluated as $\false$ and $\true$, respectively.
A function $c$ that maps to $\mathbb{B}$ is denoted a \emph{constraint} and a term of $c$ a \emph{constraint term}.
The sort $\mathbb{B}$ is available in all \smt problems by definition.
The formula $\varphi$ is \emph{satisfiable} if there exists an interpretation $I$ such that $\varphi^I$ evaluates to $\true$ and we denote any such interpretation a \emph{model} of $\varphi$.

As usual in literature, we assume that any Boolean structure in $\varphi$ is expressed in \ac{cnf}.
We denote a constraint or its Boolean negation a \emph{literal}, a (disjunctive) set of literals a \emph{clause} and treat formulas as (conjunctive) sets of clauses. We present a clause $C$ as a tuple of literals and define $\terms(C) = \bigcup_{l\in C} \terms(l)$ and $\terms(\mathcal{C}) = \bigcup_{C\in \mathcal{C}} \terms(C)$ for a set of clauses $\mathcal{C}$.

\begin{example}
	Let \[\varphi = {1}+x<f({12.7})\] be an \smt formula, with sorts propositional logic ($\mathbb{B}$), reals ($\mathbb{R}$) and integers ($\mathbb{Z}$). The interpreted function symbols are $F_{i} = \{<, +, {1}, {12.7}\}$ where $<\,: \mathbb{Z} \times\mathbb{Z} \rightarrow \mathbb{B}$ is a constraint,
	$+: \mathbb{Z}\times\mathbb{Z} \rightarrow\mathbb{Z}$ and the constants ${1}$ and ${12.7}$ are of sort integer and real, respectively.%
	\footnote{Note that, by slight abuse of notation, we do not distinguish between the (constant) term $1$ and its interpretation, the domain element $1\in\dom(\mathbb{Z})$.}
	They are all interpreted in the standard way.
	The uninterpreted function symbols are $F_{u} = \{x, f\}$ where $x$ is a function symbol of arity $0$ of sort integer, and $f:\mathbb{R} \rightarrow \mathbb{Z}$.
	Note that we mix pre- and infix notation however appropriate.
\end{example}

Given an $n$-ary interpreted function $f$, we define the \emph{evaluation function} $\eval_f$ to allow evaluating terms without a full interpretation function given.
We thus extend each sort in the signature of $f$ with the new special domain element $\udef$. Given the signature of $f$ and corresponding domain elements $s_1\in\mathcal{S}_1\cup\{\udef\}, \dots, s_n\in\mathcal{S}_n\cup\{\udef\}$, then $\eval_f(s_1,\dots,s_n) \mapsto s$ for the correct $s\in\mathcal{S}\cup\{\udef\}$ according to the interpretation of $f$.
Along with \smtlib, we require $f$ to be a total function, thus it must return a domain element if all its parameters are domain elements.
In case one or more parameters are \udef, $\eval_f$ may still return a domain value, provided the value can be computed independently of the unspecified parameters; otherwise, it returns \udef.
For a term $t$ defined by $f$, we may write $\eval_t$ for $\eval_f$.

\begin{example}
	Assume term $t = +(t_1, t_2)$ of sort $\mathbb{R}$ with $+$ as real addition. Thus, $\eval_+(4,3) = 7$, and $\eval_+(2, \udef) = \udef$.
	Let $u = \times(t_1, t_2)$ be a multiplicative term in $\mathbb{R}$ then $\eval_{\times}(0, \udef) = 0$ and $\eval_{\times}(2, \udef) = \udef$.
	Assume the constraint term $t_3 \leq 0$, then $\eval_{\leq 0}(2) = \false$.
\end{example}

%% file: rules.tex
\newcommand{\decidesto}{\rightarrowtail\xspace}
\newcommand{\propsto}{\rightarrowtail^*\xspace}
\newcommand{\just}{\ensuremath{\mathsf{just}\xspace}}
\newcommand{\feasible}{\ensuremath{\mathsf{feasible}\xspace}}
\newcommand{\infeasible}{\ensuremath{\mathsf{infeasible}\xspace}}
\newcommand{\explain}{\ensuremath{\mathsf{explain}}\xspace}
\newcommand{\consistent}{\ensuremath{\mathsf{consistent}\xspace}}
\newcommand{\satisfiable}{\ensuremath{\mathsf{satisfies}}\xspace}

\newcommand{\trail}[1]{\ensuremath{\llbracket#1\rrbracket}\xspace}
\newcommand{\emptytrail}{\trail{}}
\newcommand{\sstate}[2]{\ensuremath{\langle #1, #2\rangle}\xspace}
\newcommand{\state}[1]{\sstate{#1}{\mathcal{C}}}

To solve an \smt formula $\varphi$, we need to find an assignment to all terms in $\varphi$.
For the rest of this text, we consider $\varphi$ and the set of all terms in the search problem $T$ as fixed.

A \emph{term assignment} $t\mapsto\alpha$ assigns to a term $t\in T$ of sort $\mathcal{S}$ a value $\alpha \in \dom(\mathcal{S})$.
Given a set $S$ of term assignments, we say that $S$ is \emph{unambiguous} if for every $(t_1 \mapsto \alpha_1) \in S$ and $(t_2 \mapsto \alpha_2)\in S$, whenever $t_1 = t_2$ then $\alpha_1 = \alpha_2$.
As we only consider unambiguous sets of term assignments we can define the \emph{term assignment function} of $S$ as 
\[
\nu_S(t) = \begin{cases}
	\alpha & (t \mapsto \alpha)\in S\\
	\udef & \textrm{otherwise}
\end{cases}
\]
We say that $S$ is \emph{consistent}, denoted as $\consistent(S)$, if and only if for all $(t \mapsto \alpha)\in S$ with $t = f(t_1,\dots,t_m)$ we have that $\eval_f(\nu_S(t_1),\dots, \nu_S(t_m))$ is either $\udef$ or $\alpha$.
We define the closure $\bar{S}$ of a term assignment $S$ inductively: We start with ${\bar S} = S$; as long as there is a term $t\in T$ such that $\nu_{\bar S}(t) = \udef$ and $\eval_t(\nu_{\bar{S}}(t_1),\dots,\nu_{\bar{S}}(t_m)) = \alpha$ with $\alpha\neq\udef$, add $(t\mapsto\alpha)$ to $\bar S$.
Intuitively, the assignment $\bar S$ extends $S$ by exhaustively applying the evaluation functions until a fixed point is reached.
Note that the closure of the empty term assignment set $\bar\emptyset$ contains all term assignments induced by the background theories.
We assume that $\bar\emptyset$ is consistent.
Further note that a consistent $S$ does not guarantee $\bar S$ to be consistent.
Given a term assignment $S$ such that $\bar S$ is consistent and a yet unassigned term $t$, we denote $\feasible(S,t)\subseteq \sort(t)$ as the set of values that can be assigned to $t$ such that $\bar S$ remains consistent.
\begin{example}
	Assume a term assignment with integer terms $x$ and $y$:
	\[S = \{(x+y \teq 0) \mapsto \true,\ x \mapsto 2,\ y \mapsto 3\}\]
	then $\consistent(S)$ holds.
	Towards constructing the closure $\bar S$, the evaluation function $\eval_+$ is applied such that we get $\eval_+(2,3) = 5$ and add the term assignment $(x+y) \mapsto 5$ to $\bar S$.
	Then $\eval_\teq(x+y, 0) = \false$.
	This, however, contradicts $( (x+y \teq 0) \mapsto \true) \in S \subseteq \bar S$, therefore, $\bar S$ is not consistent.

	Consider integer term $z$ and the term assignment
	$S' = \{(z^2 \teq 1) \mapsto \true\}$,
	then both $S'$ and $\bar{S'}$ are consistent and $\feasible(S',z) = \{-1, 1\}$.
\end{example}

A \emph{trail} $M$ is a sequence of term assignments that captures the current search state.
We write $\llbracket\cdot\rrbracket$ to denote a trail and treat $M$ as a set of term assignments whenever required. By slight abuse of notation, we write $\llbracket M, t\mapsto\alpha\rrbracket$ for the extension of trail $M$ with an assignment to $t$.
Each assignment on the trail is either a \emph{decision} or a \emph{propagation}.
For a decision, the assigned value was chosen without any reasoning (usually) out of multiple options, on the other hand, a propagation assigns the only feasible value, i.e.\ $\feasible(M,t) = \{\alpha\}$, supported by theory specific reasoning.
To distinguish between propagated and decided term assignments on a trail we may write $t\propsto\alpha$ or $t\decidesto\alpha$, respectively.
The decision level of $M$ is the number of decisions on $M$. Whenever appropriate, we append a suffix to a trail to make its decision level explicit.
We may also write $t\decidesto_n\alpha$ to denote the $n$-th decision on a trail.
During the search, we aim to extend $M$ while keeping the consistency of $\bar M$.

We extend the term assignment function to evaluate clauses. Given a clause $C$, and a set of term assignments $S$, we define
\[
\nu_S(C) = \begin{cases}
	\true & \exists\, l\in C \textrm{ . } \nu_S(l) = \true \\
	\false &  \forall\, l\in C \textrm{ . } \nu_S(l) = \false \\
	\udef & \textrm{otherwise}
\end{cases}
\]
We say that a term assignment $S$ \emph{satisfies} a set of clauses $\mathcal{C}$, denoted by $\satisfiable(S,\mathcal{C})$, if all clauses are satisfied:
\begin{align*}
\satisfiable(S,\mathcal{C}) =\,& \consistent(\bar{S}) \\
                          \,\land\,& (\forall\, t\in \terms(\mathcal{C})\textrm{ . } \nu_{\bar S}(t)\neq\udef)\\
                          \,\land\,& (\forall\, C\in \mathcal{C} \textrm{ . } \nu_{\bar S}(C) = \true)
\end{align*}
We further say that $S$ is \emph{infeasible} for $\mathcal{C}$, denoted as $\infeasible(S,\mathcal{C})$, when $S$ cannot be extended such that the clauses are satisfied:
\[
\infeasible(S,\mathcal{C}) = \nexists\, S' \textrm{ . } S' \supseteq S \,\land\, \satisfiable(S',\mathcal{C})
\]
When we reach a trail $M$ such that $\satisfiable(M,\mathcal{C})$, then we can construct a model $I$ using $\bar M$ and deduce that $\mathcal{C}$ are satisfiable.%
\footnote{This is formally proven in Lemma~\ref{lem:model-existence} of the appendix (Section~\ref{sec:proofs}).}

\subsection{Theory Reasoning and Explanations}\label{sec:pluginExpl}
\paragraph{Theory Plugins.}
The general \mcsat search is independent of any background theory. As mentioned above, the only sort natively supported in all \mcsat style search procedures is the Boolean sort $\mathbb{B}$.
In order to add an \smtlib theory including its sorts and interpreted function symbols, a \emph{theory reasoning plugin} $\mathcal{T}$ is attached to the \mcsat search.
A plugin $\mathcal{T}$ has a set of sorts $\sort(\mathcal{T})$ which values it can interpret as well as a set of functions $\func(\mathcal{T})$ on those sorts.
We say that a plugin can \emph{handle} any sort that is has and \emph{supports an interpreted function} $f$ if it provides $\eval_f$.
While more than one plugin may be able to handle a particular sort, we require all plugins supporting an interpreted function to agree on its evaluation.
A plugin $\mathcal{T}$ \emph{supports an uninterpreted function} $u$ by ensuring congruence, i.e.\ that two terms $t_1, t_2$ of $u$ are assigned the same value if their sub-terms are assigned the same values.
Usually there is one dedicated plugin to ensure congruence for all uninterpreted functions (see Section~\ref{sec:pluginUf}).
Each interpreted or uninterpreted function must be supported by at least one theory plugin to ensure congruence and correct interpretation.
We further say that a plugin \emph{supports a term} if it is constructed by a supported function.

Assume a plugin $\mathcal{T}$ and a supported term $t$ that contains a non-supported sub-term $t'\in\subterms(t)$. Then $\mathcal{T}$ can handle $\sort(t)$. For $\mathcal{T}$, $t'$ is an atomic term that cannot be evaluated any further, but it can assign a value to $t'$ and use any assigned value to evaluate $t$.
We say that $t'$ is a \emph{variable} for $\mathcal{T}$.
If a supported term $t$ contains a non-supported $t'\in\subterms(t)$, $\mathcal{T}$ treats $t'$ as an atomic term, the same way as a constant that cannot be further evaluated.
An uninterpreted atomic term is considered a \emph{variable} for the respective plugin.
It may assign values to $t'$ as for any constant and must handle assignments to $t'$ by a different plugin.
\begin{example}\label{ex:termstructure}
	\definecolor{colorB}{RGB}{0, 0, 0}
	\definecolor{colorUF}{RGB}{230, 159, 0}   %
	\definecolor{colorNRA}{RGB}{0, 158, 115}  %
	
	Consider the formula $\varphi = [f(x+y) + 2y \leq 0] \land [f(0) - y \teq 0]$ where $x,y$ are of sort integer~$(\mathbb{Z})$.
	The term structure of $\varphi$ is depicted below. Rounded corners indicate sort $\mathbb{B}$, all other terms are of sort $\mathbb{Z}$.
	Colors indicate the responsible plugin for each term, where \textcolor{colorNRA}{green} and \textcolor{colorUF}{orange} indicate the plugins of integer arithmetic $\mathcal{I}$ and uninterpreted functions $\mathcal{U}$, respectively.
	Boolean terms, which are handled by the \mcsat reasoning procedure directly, are depicted black.

\begin{center}
\begin{forest}
for tree={
	s sep=4mm,
	l sep=0mm,
	l=6mm,
	draw,
	inner sep=1.5pt,
	minimum size=0pt,
	outer sep=0pt
}
[$\land$, draw=colorB, rounded corners=2pt
[${\leq\,}0$, edge={colorB}, draw=colorNRA, rounded corners=2pt, font=\scriptsize
[$+$, edge={colorNRA}, draw=colorNRA
[$f$, edge={colorNRA}, draw=colorUF
[$+$, edge={colorUF}, draw=colorNRA
[$x$, edge={colorNRA}, draw=colorNRA]
[$y$, edge={colorNRA}, draw=colorNRA]
]
]
[$*$, edge={colorNRA}, draw=colorNRA
[$2$, edge={colorNRA}, draw=colorNRA]
[$y$, edge={colorNRA}, draw=colorNRA]
]
]
]
[${\teq\,}0$, edge={colorB}, draw=colorNRA, rounded corners=2pt, font=\scriptsize
[$-$, edge={colorNRA}, draw=colorNRA
[$f$, edge={colorNRA}, draw=colorUF
[$0$, edge={colorUF}, draw=colorNRA]
]
[$y$, edge={colorNRA}, draw=colorNRA]
]
]
]
\end{forest}
\end{center}
Consider the term $f(x+y)$ which is of sort $\mathbb{Z}$. It is a variable for $\mathcal{I}$, thus the plugin assigns and reads the term's (integer) value like any constant.
The $\mathcal{U}$ plugin can do the same and ensures the congruence property of $f$ when its sub-term $x+y$ has a value assigned.
The same happens for the unary constraint terms ($\leq\! 0$ and $\teq\! 0$) which are variables in \mcsat's $\mathbb{B}$ reasoning. $\mathcal{I}$ can set their boolean value (by evaluating the constraints using their evaluation function) or the \mcsat search can force them to be either $\true$ or $\false$ using Boolean reasoning, requiring $\mathcal{I}$ to set values accordingly.

Note that in this example constraints have been normalized as comparisons to $0$ and defined as unary constraints. This is common in computer algebra and, of course, a possible alternative formalization in an \mcsat procedure that is actually implemented in \yices.
However, for better readability, we are using binary comparison constraints in other examples.
\end{example}

\paragraph{Explanations.}
While different plugins interact with each other using the assigned terms on the trail, they need to translate theory-specific reasoning for those assignments into the common \mcsat search procedure. This is done using a plugin-specific \emph{explanation function}.
Given a formula $\varphi$ as the search problem, a trail $M$, and a (constraint) term $p$ of sort $\mathbb{B}$, then lemma clause $E$ is constructed by
\[ E = (e_1,\dots,e_n, p) = \explain(M, p) \]
fulfilling the following conditions:
\begin{enumerate*}[label=(\roman*)]
	\item \label{exp:valid} $E$ is a valid lemma in the search problem, i.e. for every interpretation $I$ such that $\varphi^I=\true$, there exists an $e\in E$ such that $e^I = \true$.%
	\footnote{Note that this is a weaker definition than presented in previous literature which requires a valid lemma in \textbf{any} interpretation, i.e. for any interpretation $I$, there exists an $e\in E$ such that $e^I = \true$.}
	\item \label{exp:just} $E$ justifies $p$, i.e. for every $e_i$ with $1\leq i \leq n$ we have $\nu_{\bar M}(e_i) = \false$.
	\item \label{exp:fb}   $E$ may introduce new literals, as long as they stem from a finite basis.
\end{enumerate*}
The implementation of $\explain$ is fully dependent on the theory plugin that produces the lemma.
In fact, admitting an $\explain$ function is all a background theory must support for integrating it in an \mcsat search (c.f.~\cite{mcsat}).
Note that the explanation function might generate valid theory lemmas (independent of $\varphi$), thus ensuring for any interpretation $I$, there exists $e\in E$ such that $e^I=\true$.
\begin{example}\label{ex:explain}
	Assume a search in real arithmetic with real terms $x$ and $y$ reaches a trail
	\[M = \llbracket (xy \leq 0) \mapsto\true, x \mapsto 1, y\mapsto 1 \rrbracket\]
	which clearly leads to a not consistent $\bar M$.
	An explanation function $E = \explain(M, \bot)$ is used to generate an explanation clause that translates this arithmetic reasoning to a clause:
	\begin{align*}
	E =& \left( (xy\leq 0) \,\land\, (x > 0) \,\land\, (y > 0) \implies \bot \right) \\
	  =& \left( (xy\leq 0) \,\land\, (x > 0) \implies (y \leq 0) \right) \\
	  =& \left( \lnot (xy\leq 0) \,\lor\, \lnot (x > 0) \,\lor\, (y \leq 0) \right)
	\end{align*}
	Clearly, this fulfills the requirements of an explanation function. It is a valid lemma (requirement \ref{exp:valid}), as it is true for any assignment to $x$ and $y$ which can be seen by trivial algebraic reasoning.
	It justifies $\bot$ on the trail $M$ (requirement \ref{exp:just}), as $\lnot (xy\leq 0)$ is on the trail in opposite polarity, and $\nu_{\bar M}(\lnot (x>0))$ as well as $\nu_{\bar M}(y \leq 0)$ are \false.
	The finite basis condition \ref{exp:fb} is a property of the explanation procedure as a whole.
	The explanation introduced the new terms $x>0$ and $y\leq 0$ and these terms must come from a finite set of terms the procedure can generate.
	For example, it must not generate arbitrary many terms like $\{x>0,\, x>1,\, x>2,\,\dots\}$ as this would violate termination.
\end{example}

During the search the explain function is used in two cases:
\begin{enumerate*}[label=(\roman*)]
	\item Whenever theory-specific procedures detect that $\bar M$ is inconsistent, then $C = \explain(M,\bot)$ is used to derive the \emph{conflict clause} $C$. By the construction requirements of the explanation function, this excludes the current trail and forces the search to pivot.
	\item When propagating a term assignment, i.e. when a theory related procedure detects that $\feasible(M,t) = \{\alpha\}$ for a yet unassigned term $t$, we use $J = \explain(M, t = \alpha)$ to generate a \emph{justification} $J$ that explains the propagation with regard to the current trail $M$. By construction of $J$, we have that $(t = \alpha)\mapsto\true$ on $M$ and, thus, $(t \propsto\alpha)$ is added to $M$. We denote $J = \just(t\propsto\alpha)$ and write $J = \just(t)$ as a shorthand.
	It must be available whenever required by conflict resolution as long as the propagation is on the trail.
	This enables justifications to be calculated lazily, i.e.\ only when required by conflict analysis. This is important, as calls to the \explain\ function might be costly, depending on the theory.
\end{enumerate*}

While \mcsat performs all Boolean propagations, any propagation for other sorts is optional.
Propagations are helpful for performance in general, constructing the required justification may be infeasible in many cases and a plugin is not required to perform a possible propagation.
In case a propagation of term $t$ to value $\alpha$ is detected, but no explanation can be constructed (with reasonable effort), it is, in general, still beneficial to add $t\decidesto\alpha$ to $M$ as a decision (c.f.~\cite{cav25}).

\subsection{Search Rules}

As in the previous publications, we are describing the \mcsat search as a system of symbolic proof rules.
Given a set of clauses $\mathcal{C}$, we start with an empty trail $M = \emptytrail_0$ and end in one of the two terminal states \ssat or \susat.

An active search can be in one of two states: \emph{search state} or \emph{conflict state}.
While in a search state, $M$ gets extended with either decisions or propagations until everything is assigned or $\bar M$ is detected to be inconsistent.
In a conflict state, we aim to recover from an inconsistent $\bar M$ or detect the unsatisfiability of $\mathcal{C}$.
We explain the conflict as a conflict clause $C$, resolve $C$ with previous propagations and remove assignments from $M$ until a relevant decision is reached.
In the rules below, we denote a search state as a pair $\state{M}$ of trail $M$ and a set of clauses $\mathcal{C}$.
A state in conflict mode with conflict clause $C$ is written as $\state{M} \vdash C$.
We denote rules that take a conflict state as \emph{conflict rules}, all other rules as \emph{search rules}.

\begin{runningexample}
	Let $\mathcal{C} = \{C_1, C_2, C_3, C_4, C_5\}$ be an instance over propositional logic $\mathbb B$ and real arithmetic $\mathbb R$.
	\setlength{\abovedisplayskip}{0pt}
	\setlength{\belowdisplayskip}{0pt}
	\begin{align*}
		C_1 &= (x^2 \teq 4 \lorS y \teq 8) & C_4 &= (x > 5 \lorS x + y \teq 1) \\
		C_2 &= (x^2 + 2y > y^2 \lorS y \teq 1) & C_5 &= (x \teq 7y \lorS x \tneq y + 3) \\
		C_3 &= (1 < -x \lorS (x - 2)y \tneq 3)
	\end{align*}
	Here $x$ and $y$ are uninterpreted constants of sort $\mathbb{R}$.
	In this running example we denote the number of applied rules in the superscript of the trail.
	We start the search in the state $\langle M^0_0, \mathcal{C}\rangle$.
	By definition, we have $M^0_0 = \llbracket \rrbracket_0$.
	For brevity we write $p$, $p_n$, and $p^*$ on the trail instead of $p\decidesto\true$, $p\decidesto_n\true$, and $p\propsto\true$, respectively.
\end{runningexample}

\newcommand{\becomes}{\ensuremath{\quad\longrightarrow\quad}}

\begin{rulebox}{Decide}
Assume a trail $M_n$ and a term $x$ such that $\nu_{\bar M}(x) = \udef$, then any value $\alpha \in \feasible(M, x)$ can be decided for $x$.
\[\state{M_n} \becomes \state{\llbracket M, x\decidesto_{(n+1)}\alpha\rrbracket}\]
\end{rulebox}
While a general term can be decided upon, in practice, deciding is limited to variables, i.e. terms that are atomic for a dedicated theory plugin. All non-atomic terms can be evaluated using the plugin's $\eval$ function and, thus, will get an assignment using the closure of the trail $\bar M$.
Although more than one theory plugin can support $\sort(t)$, we usually dedicate one plugin to make decisions for each particular sort.
While this rule does not limit the choice for $x$ and $\alpha$, in any practical implementation, those choices are important considerations for the performance of the solver. Heuristics are essential to achieve reasonable performance in practice~\cite{cav25,l2o}.

\begin{runningexample}
	In the first search state $\langle M^0_0,\mathcal{C}\rangle$, we decide $x \decidesto_1 2$ using \nameref{rule:Decide}:
	\[\langle \llbracket x \decidesto_1 2 \rrbracket^1_1,\mathcal{C}\rangle\] and denote the new trail $M_1^1$.
	Note that in the new search state $\langle M^1_1,\mathcal{C}\rangle$ we cannot decide on the propositional variable $(x - 2)y \tneq 3$, because $\nu_{\bar {M^1_1}}((x - 2)y \tneq 3) = \true$,
	but on $x^2 + 2y > y^2 \decidesto_2 \true$ to get:
	\[\langle \llbracket x \decidesto_1 2, (x^2 + 2y > y^2)_2\rrbracket^2_2,\mathcal{C}\rangle\]
\end{runningexample}

\begin{rulebox}{Propagate}
Assume a clause of the form $C = \explain(M, x = t) = (l_1\land\dots\land l_\ell \implies x = t) = (\lnot l_1 \lor \dots \lor \lnot l_\ell \lor x = t)$ with a term $t$ such that $\nu_{\bar M}(t) = \alpha\in\sort(t)$ can be derived from the clause database.
By construction, for all $1\leq i \leq \ell$ we have that $\nu_{\bar M}(l_i) = \true$, thus we can set $\alpha$ for $x$.
\[\state{M_n} \becomes \state{\llbracket M, x\propsto\alpha\rrbracket_n}\]
\end{rulebox}
Note that not all possible propagations are actually required to be performed by a theory plugin.
Even if a plugin realizes that there is only one feasible value for a term, it may decide not to propagate the value, because a justification is either too expensive or even impossible to calculate.
In such cases, it is beneficial to force a decision of that variable soon (c.f.~\cite{cav25}).
For completeness of the system, propagations are not required, as the system could also decide an assignment to the propagated value.
For infinite domains, it might even be impossible to generate any justification as this would violate the finite basis requirement of the \explain function (c.f.\ Example~\ref{ex:explain}.
\begin{runningexample}
	We continue with $M_2^2$ as the current trail.
	The clause $C_3$ has only one literal that does not evaluate to $\false$. 
	So we use Boolean unit propagation to get the new state: \[
	\langle\llbracket x \decidesto_1 2, (x^2 + 2y > y^2)_2, ((x - 2)y \tneq 3)^* \rrbracket^3_2, \mathcal{C}\rangle
	\]
	This is possible because we take $\explain(\langle M^2_2, \mathcal{C}\rangle, (x - 2)y \tneq 3) = C_3 = (1 < -x \lorS (x - 2)y \tneq 3) = (\lnot(1 < -x) \limpS (x - 2)y \tneq 3)$ and $\nu_{\bar{M^2_2}}(\lnot(1 < -x)) = \true$, so we can propagate $(x - 2)y \tneq 3 \propsto \true$ with justification $C_3$.
	We can perform the similar propagation $x + y \teq 1 \propsto \true$ with $\just(x + y \teq 1) = C_4 = (x > 5 \lorS x + y \teq 1)$ and get: \[\langle\llbracket x \decidesto_1 2, (x^2 + 2y > y^2)_2, ((x - 2)y \tneq 3)^*, (x + y \teq 1)^* \rrbracket^4_2,\mathcal{C}\rangle\]
	Finally, we perform the arithmetic theory propagation $y \propsto -1$ with justification 
	$ \explain(\trail{\dots}, y \teq -1) = (x \tneq 2 \lorS \lnot(x + y \teq 1) \lorS y \teq -1)$ and get:
	\[
	\langle\llbracket x \decidesto_1 2, (x^2 + 2y > y^2)_2, ((x - 2)y \tneq 3)^*,(x + y \teq 1)^*, y \propsto -1 \rrbracket^5_2,\mathcal{C}\rangle
	\]
\end{runningexample}

\begin{rulebox}{Conflict}
Given $\infeasible(M,\mathcal{C})$, let $E = \explain(M,\bot) = (e_1\land\dots\land e_\ell \implies \bot) = (\lnot e_1\lor\dots\lor\lnot e_\ell)$ such that for all $1\leq i \leq \ell$ we have $\nu_{\bar M}(e_i) = \false$, then the search switches to the conflict mode with $E$ as its initial conflict clause:
\[\state{M_n} \becomes \state{M_n} \vdash E\]
\end{rulebox}
\begin{runningexample}
	We continue in the search state $\langle M^5_2,\mathcal{C}\rangle$, with $M^5_2 = \llbracket x \decidesto_1 2,$ $ \dots, y \propsto -1 \rrbracket_2$.
	Note that $C_5 = (x \teq 7y \lorS x \tneq y + 3) \in \mathcal{C}$, so we get $\nu_{\bar{M^5_2}} (C_5) = \false$, which is a propositional conflict.
	Therefore, we use rule \nameref{rule:Conflict} with $E = \explain(\langle M^5_2, \mathcal{C} \rangle, \bot) = C_5$ to get (with same trail $M^6_2 = M^5_2$):
	\[ \langle M_2^6,\mathcal{C}\rangle \vdash C_5 \]
\end{runningexample}

\begin{rulebox}{Resolve} When in conflict mode, we can resolve a previous made propagation using its justification if the propagated term $x$ occurs in the conflict clause $C$, i.e. $x\in \terms(C)$:
	\[\state{\llbracket M, x\propsto \alpha\rrbracket_n} \vdash C \becomes \state{M_n} \vdash R \]
	Then by the construction of the propagation, the justification is of the form $\just(x) = (l_1\land\dots\land l_\ell\implies x \teq t)$ and the resolved conflict clause $R = C[x/t] \cup (\lnot l_1\lor\dots\lor \lnot l_\ell)$.
\end{rulebox}
Before continuing with the running example, let us first consider two other applications of \nameref{rule:Resolve}:
\begin{example}
	Assume $\mathbb{B}$ variables $a$, $b$, and $c$, the clause $D = (a \lor c)$, and a trail that assigns $a$ to $\false$. Then the propagation rule can use $D$ for the propagation $\lnot a \implies c$ (by simplifying $c \teq \true$).
	Later, a conflict clause $C = (b \lor \lnot c)$ becomes $R = (b \lor a)$, which is Boolean resolution.
\end{example}
\begin{example}
	Let $y$ and $z$ be integer variables, $D = (y < 3 \lor z \teq 5)$ a clause, and a trail that assigns $y$ to~$6$. Then the propagation rule can use $D$ for the propagation $\lnot (y < 3) \implies z \teq 5$.
	Later a conflict clause $C = (y \cdot z \teq 3 \lor \lnot (z \teq 5))$ becomes $R = (y \cdot 5 \teq 3 \lor y < 3)$, which is arithmetic resolution.
\end{example}

\begin{runningexample}
	In the conflict state \[
	\langle\llbracket x \decidesto_1 2, (x^2 + 2y > y^2)_2, ((x - 2)y \tneq 3)^*,(x + y \teq 1)^*, y \propsto -1 \rrbracket^6_2,\mathcal{C}\rangle \vdash C_5
	\]
	we can apply the rule \nameref{rule:Resolve}, because $y \in \terms(C_5)$.
	Recall $\just(y) = (x \tneq 2 \lorS \lnot(x + y \teq 1) \lorS y \teq -1)$ so the resolvent is $R = (x \teq 7y \lorS x \tneq y + 3)[y/-1]\cup(x \tneq 2 \lorS \lnot(x + y \teq 1)) = (x \teq -7 \lorS x \tneq 2 \lorS \lnot(x + y \teq 1))$.
	This yields the new conflict state:
	\[
	\langle \llbracket x \decidesto_1 2, (x^2 + 2y > y^2)_2, ((x - 2)y \tneq 3)^*, (x + y \teq 1)^*\rrbracket^7_2, \mathcal{C}\rangle \vdash R
	\]
	Because $(x + y \teq 1) \in R$, we can perform a similar resolution on $x + y \teq 1 \propsto \true$, which has the justification $C_4 = (x > 5 \lorS x + y \teq 1)$.
	So $R' = $ $R[x + y \teq 1/\top]\cup(x > 5) = $$ (x \teq -7 \lorS x \tneq 2 \lorS \lnot(\top) \lorS x > 5) = (x \teq -7 \lorS x \tneq 2 \lorS x > 5)$.
	The new conflict state is: \[
	\langle \llbracket x \decidesto_1 2, (x^2 + 2y > y^2)_2, ((x - 2)y \tneq 3)^* \rrbracket^8_2, \mathcal{C} \rangle \vdash R'
	\]
	Note that for the propagation $(x - 2)y \tneq 3 \propsto \true$, with justification $(1 < -x \lorS (x - 2)y \tneq 3)$, the \nameref{rule:Resolve} rule cannot be applied, because $((x - 2)y \tneq 3) \notin \terms(R')$.
\end{runningexample}

\begin{rulebox}{Consume} In conflict mode, whenever a propagation is not part of the conflict clause, it may be omitted. Assume $x\notin \terms(C)$, then
\[\state{\llbracket M, x\propsto \alpha\rrbracket_n} \vdash C \becomes \state{M_n} \vdash C \]
\end{rulebox}
\begin{runningexample}
	In the current conflict state, we can use rule \nameref{rule:Consume} to get rid of the last trail element and reach state
	\[
	\langle \llbracket x \decidesto_1 2, (x^2 + 2y > y^2)_2 \rrbracket^9_2, \mathcal{C} \rangle \vdash R',
	\]
\end{runningexample}

\begin{rulebox}{Drop}
When a decision does not influence the current conflict clause $C$, then the decision can be dropped.
Assuming $\nu_{\bar{M_n}}(C) = \false$, any decision of level $n+1$ can be dropped:
\[\state{\llbracket M_{n}, x\decidesto_{n+1} \alpha\rrbracket} \vdash C \becomes \state{M_n} \vdash C \]
\end{rulebox}
In practice, one can calculate the decision level $n$ of $C$ and drop all decisions and corresponding propagations with level greater than $n$ at once.
\begin{runningexample}
	In the conflict state \[
	\langle \llbracket x \decidesto_1 2, (x^2 + 2y > y^2)_2 \rrbracket^9_2, \mathcal{C} \rangle \vdash (x \teq -7 \lorS x \tneq 2 \lorS x > 5),
	\]
	the rule \nameref{rule:Drop} is applied because $\nu_{\bar {M'}}(x \teq -7 \lorS x \tneq 2 \lorS x > 5) = \false$ with $M' = \trail{x \decidesto 2}$:
	\[
	\langle \llbracket x \decidesto_1 2\rrbracket^{10}_1,\mathcal{C}\rangle \vdash (x \teq -7 \lorS x \tneq 2 \lorS x > 5)
	\]
	We cannot apply the \nameref{rule:Drop} further as $\nu_{\bar\emptyset}(x \teq -7 \lorS x \tneq 2 \lorS x > 5) = \udef$.
\end{runningexample}

\begin{rulebox}{Backjump}
When a decision on the trail influences the conflict clause $C$, then undoing the decision changes $\nu_{\bar{M_n}}(C)$ to $\udef$.
Then there exists a subset of literals $C' \subseteq C$ such that $\nu_{\bar{M_n}}(L) = \udef$ for all $L\in C'$.
In case $C' = \{L\}$ then $L$ can be propagated by $C$. Thus a propagation is put on the trail with $\just(L) = C$:
\[\state{\llbracket M_{n}, x\decidesto_{n+1} \alpha\rrbracket} \vdash C \becomes \sstate{\llbracket M_n, L\propsto \true \rrbracket}{\mathcal{C}} \]
Otherwise, we need to perform a Boolean decision and pick a choice $L\in C'$ to satisfy $C$:%
\[\state{\llbracket M_{n}, x\decidesto_{n+1} \alpha\rrbracket} \vdash C \becomes \sstate{\llbracket M_n, L\decidesto_{n+1} \true \rrbracket}{\mathcal{C}} \]
When choosing $L$ we need to make sure that $x\in\terms(L)$, i.e. $x\neq L$.
\end{rulebox}
The last condition for the second case is important to ensure termination.
In case $x$ is a literal as well as a sub-term of another literal (for example if $x$ is a propositional term and $C = \{x,\, f(x)\}$), we need to ensure not to replace a decision on $x$ with another decision on $L$ (indefinitely).
\begin{runningexample}
	We are in the conflict state: \[
	\langle \llbracket x \decidesto_1 2\rrbracket^{10}_1,\mathcal{C}\rangle \vdash (x \teq -7 \lorS x \tneq 2 \lorS x > 5)
	\]
	Here the decision $x \decidesto_1 2$ affects the evaluation of the conflict clause $R'$ because $\nu_{\bar\emptyset}(R') = \udef$.
	Before leaving the conflict state, we apply the \nameref{rule:Learn} rule (see below) in order to add the final resolvent $R'$ to the clause set.
	Then we can apply the \nameref{rule:Backjump} to leave the conflict state.
	Since multiple literals of $R'$ evaluate to $\udef$, we replace the decision on the trail with one of the following decisions in $R'_u = R' = \{(x \teq -7)_1, (x \tneq 2)_1, (x > 5)_1\}$ and pick $x \teq -7 \decidesto_1 \true$.
	Then the new state is
	\[
	\langle\llbracket (x \teq -7)_1 \rrbracket^{11}_1, \mathcal{C}\cup\{R'\}\rangle
	\]
\end{runningexample}

\begin{rulebox}{Sat}
Assume a trail $M$ such that $\satisfiable(M,\mathcal{C})$, then
\[\state{M} \becomes \mathsf{sat}\]
\end{rulebox}
Note that this is also the case when for each term $t$, we have that $\nu_{\bar M}(t) \neq \udef$ and not $\infeasible(M)$.
In the \ssat case, the model can be constructed by taking the closure of the last trail and assign all still unassigned terms with arbitrary consistent values.
\begin{runningexample}
	We propagate $x\propsto -7, x + y \teq 1 \propsto \true$ and $y \propsto 8$ to get:
	\[
	\langle \llbracket (x \teq -7)_1, x \propsto -7, (x + y \teq 1)^*, y \propsto 8 \rrbracket^{14}_1, \mathcal{C} \cup \{R'\}\rangle
	\]
	The Sat rule can be applied here because $\satisfiable(M_1^{14},\mathcal{C}\cup\{R'\})$ holds.
\end{runningexample}

\begin{rulebox}{Unsat}
Assume a trail $M$ and a conflict at decision level $0$, then
\[\state{M_0} \vdash C \becomes \mathsf{unsat}\]
\end{rulebox}

\subsection{Learning and Forgetting}

The rules presented above are sufficient to perform an \mcsat search.
For practical performance reasons it is, however, beneficial to add, i.e. learn, clauses to the clause database $\mathcal{C}$ as well as delete, i.e. forget, previously learned clauses.
Note that any conflict clause $C$ is logically implied by the clause database $\mathcal{C}$.
This is because any explanation clause is implied by $\mathcal{C}$ by requirement \ref{exp:valid} of \explain\ and the resolvent of two implied clauses with rule \nameref{rule:Resolve} keeps this property.%
\footnote{This is formally proven in Lemma~\ref{lemma:clausesImplied} of the appendix (Section~\ref{sec:proofs}).}
Adding or removing an implied clause to a set of clauses does not change the models.

There are two opportunities to learn clause:
\begin{enumerate*}[label=(\roman*)]
	\item any clause during conflict analysis, i.e. the original conflict clause or any resolvent thereof,
	as well as
	\item any justification that was generated during a propagation.
\end{enumerate*}
In practice, the most important clause to learn tends to be the last conflict clause right before exiting conflict analysis with rule \nameref{rule:Backjump} as this clause is used for pivoting the search.

\begin{rulebox}{Learn}
	Assume a search state in conflict with clause $C$, then $C$ can be added to the clause database $\mathcal{C}$ if $C\notin\mathcal{C}$: 
	\[\state{M} \vdash C \becomes \sstate{M}{\mathcal{C}\cup\{C\}}\vdash C\]
\end{rulebox}

\begin{rulebox}{Learn-Justification}
	Assume the search propagated $x\propsto \alpha$, then $J = \just(x\propsto \alpha)$ can be added to the clause database $\mathcal{C}$ if $J\notin\mathcal{C}$: 
	\[\state{\llbracket M, x\propsto \alpha\rrbracket} \becomes \sstate{\llbracket M, x\propsto \alpha\rrbracket}{\mathcal{C}\cup\{J\}} \]
\end{rulebox}
If, for practical performance considerations, the justification $J$ is calculated lazily, i.e. only when required by conflict analysis, it can still be learned upon calculation.
Since, in the proof system, we assume that justifications are available right away, we limit applications of \nameref{rule:Learn-Justification} to be used directly after \nameref{rule:Propagate} for technicalities the termination proof.

\begin{rulebox}{Forget} Assume clause $C\in\mathcal{C}$ is a learned clause, then we may forget $C$:
	\[\state{M} \becomes \sstate{M}{\mathcal{C}\setminus\{C\}}\]
\end{rulebox}
The \nameref{rule:Forget} rule is optional but important for practical performance.
As for any choice in the search, heuristics can greatly help in deciding which clauses to forget.

It is important to note that a justification $J = \just(x)$ of a propagation $x\propsto\alpha$ on the trail can not be forgotten, as $\just(x)$ is a property of the trail element and not an element of $\mathcal{C}$.
If $J$ had previously been leaned using \nameref{rule:Learn-Justification}, it can be removed from $\mathcal{C}$, however $J$ is available as $\just(x)$ as long as $x\propsto\alpha$ is remains on the trail.

\subsection{Soundness, Completeness, and Termination} 

The main theorem of~\cite{mcsat} is the soundness and completeness of the calculus.
For our adapted set of rules, we provide the same result and give the proof in the appendix (Section~\ref{sec:proofs}).
\begin{theorem}\label{thm:term}
	Given a set of clauses $\mathcal{C}$, and assuming an explanation function $\explain$, any derivation starting from the initial state $\state{\emptytrail}$ will terminate either in a state $\mathsf{sat}$, when $\mathcal{C}$ is satisfiable (soundness), or in the state $\mathsf{unsat}$, when $\mathcal{C}$ is unsatisfiable (completeness).
\end{theorem}

\undef\becomes

%% file: plugins.tex
In Section~\ref{sec:pluginExpl} we have outlined the requirements of a theory reasoning plugin and a corresponding explanation function.
Based on these definitions we will provide some examples for specific theories.

\paragraph{Propositional Logic.}
Although not an actual theory plugin, we can define an explanation function for a Boolean conflict.
Recall that the Boolean sort $\mathbb{B} = \{\true, \false\}$ together with the Boolean connectives translated in \ac{cnf} are natively handled by \mcsat.
Given a clause database $\mathcal{C}$ and a trail $M$, a conflict is a clause $C \in \mathcal{C}$ that evaluates to $\false$ under $M$.
The explanation function $\explain(M, \bot)$ simply returns the clause $C$.
$C$ trivially justifies $\bot$ and does not introduce new literals, thus satisfying property \ref{exp:just} and \ref{exp:fb} of an explanation function.
Although $C$ is not a valid lemma in general, it is (trivially) implied by the clause database $\mathcal{C}$ and thus it is a lemma within the search problem.

The same works for propagation. The propagate rule can be used for any clause $C\in\mathcal{C}$ that is unit with regard to $M$. The explanation is simply $C$.
\paragraph{Arithmetic.}
Non-linear real arithmetic has been the first theory to be applied to \mcsat in its original \nlsat publication~\cite{nlsat}.
It is based on \ac{cad} which is a well known technique from computer algebra. The theory can solve polynomial constraints with integer coefficients and defines the sort of reals $\mathbb{R}$ containing all algebraic numbers, i.e. the subset of real numbers that can be the root of a polynomial with integer coefficients.
Note that $\mathbb{R}$ is of infinite size.
To gain a finite basis $\explain$ function, the theory plugin uses \ac{cad} to describe an $n$-dimensional area in which all polynomials of the conflict are sign-invariant.
As all possible decompositions of the (finitely many) input polynomials contain only finitely many such areas, there are only finitely many terms to describe those areas.
For further details on the implementation of this theory we refer to~\cite{nlsat} as well as to the SMT solver \smtrat~\cite{smtrat} which contains one of the most advanced implementations of this \mcsat theory.
In addition to reals, there haven been \mcsat theories for linear integers~\cite{mcsat_inear_int} as well as non-linear integers~\cite{DBLP:conf/vmcai/Jovanovic17}.

\paragraph{Uninterpreted Functions.}\label{sec:pluginUf} The theory of uninterpreted functions (UF) does not contain its own sort.
It rather tracks the assignments of its terms by other theories plugins and detects conflicts and possible propagations using equality graphs (e-graphs).
Therefore, each generated explanation $E$ given by $\explain$ only contains existing literals, trivially fulfilling property~\ref{exp:fb} of an explanation function.
$E$ contains terms from the trail that are equal due to the e-graph, by the correctness of the e-graph, this satisfies \ref{exp:valid}.
The third property \ref{exp:just} is fulfilled by the construction of $E$, as only assigned terms are considered.
Based on the theory UF, an \mcsat implementation of the theory of arrays has been implemented~\cite{DBLP:conf/smt/IrfanG24}.

\paragraph{Further theories.} Further interesting theories that have been implemented in an \mcsat search are the theory of fixed-size bit-vectors (BV)~\cite{mcsat_bv} and the theory of finite fields (FF)~\cite{DBLP:conf/ijcar/HaderKIGK24,DBLP:conf/smt/HaderO24}.

%% file: conclusion.tex
\paragraph{Implementation Notes.}
The general idea of this adapted version of the \mcsat calculus was based on its implementation within \yices~\cite{mcsat2}.
The software architecture of the \mcsat engine includes a general \mcsat core independent of any theory.
Each theory (including the theory of propositional logic) is implemented as a plugin that provides, using a well-defined interface, functionality for decisions, propagations as well as conflict detection and explanation.
Each plugin can register terms and types that it is capable of handling. This lead us to defining the transition systems on term assignments instead of a strict separation in theory and propositional variables.
The core \mcsat engine queries all plugins for all available propagations before proceeding with the next decision. It is up to the individual plugin to choose which propagations to perform. See~\cite{mcsat2,cav25} for a description of the general \mcsat core algorithm in \yices.

While the \mcsat proof system requires an explanation function to be finite basis for proving termination, the implementation in \yices departs from this requirement: by letting the arithmetic variable order for decisions change dynamically using the \ac{vsids} heuristic, explanations generated by \ac{cad} the finite basis cannot be guaranteed.
Although, this is a bad choice for termination, \ac{vsids} improves performance for practical instances significantly~\cite{cav25}.

\paragraph{Calculus Updates. }
Notable changes to the \mcsat proof calculus as presented in this work compared to the original publications~\cite{nlsat, mcsat, mcsat2} are:
\begin{itemize}
	
	\item \textit{Term assignments instead of theory variables.}\quad
	Previous publications defined two distinct kinds of variables: propositional literals and theory variables.
	The former were any constraints over a term in a theory, the latter an uninterpreted constant in a particular theory.
	An (overloaded) \textsf{value} function evaluates literals according to assignment to theory variables.
	We changed this to better match the definitions of functions and sorted terms in \smtlib and allow an arbitrary assignments to terms.
	This enables us to present theory combination in \mcsat closer to implementation, especially in the presence of uninterpreted functions (see Example~\ref{ex:termstructure}).
	
	\item \textit{Combining Boolean and theory reasoning.}\quad 
	Due to the clear distinction between theory variables and propositional literals, previous publications maintained two sets of proof rules
	to cater to both cases.
	This led to similar rules with minor differences, depending on whether they are handling theory terms or the propositional structure.
	In this work, we implemented Boolean reasoning closer to theory reasoning to allow a single set of rules to handle both cases (see, for example, rule~\nameref{rule:Propagate} or rule~\nameref{rule:Decide}).
	While Boolean reasoning is still a special case due to additional rules for conflict analysis, 
	this clarifies the meaning of the rules and avoids minor differences in similar rules.
	
	\item \textit{Common framework for propagation.}\quad
	The original publications of \mcsat did not provide a way to propagate theory values as depicted in rule~\nameref{rule:Propagate}.
	This form of propagation was introduced as a side note in a later publication~\cite{jovanovic2021interpolation}.
	In this work, the \nameref{rule:Propagate} and \nameref{rule:Resolve} rules have been adapted to support both theory and Boolean propagation.
	By adding the \nameref{rule:Learn-Justification} rule, learning a potentially expensive justification for a theory-based propagation has been enabled, which was not possible by \nameref{rule:Learn}.

	\item \textit{Updated \explain function.}\quad
	The original publications require the \explain function to return a valid lemma $E$. A valid lemma is a clause that is true under any interpretation, i.e. $\vDash E$.
	We relax this requirement to $E$ to be true under any model of the search problem $\varphi$, i.e., $\varphi\vDash E$.
	This is a strictly weaker definition, as any valid explanation still fulfills the requirement but allows modeling Boolean conflicts as regular explanations.
	Furthermore, this change gives more flexibility to the implementation of the \explain function for particular theories, for example, allowing it to return an empty clause if the theory's reasoning procedure detects unsatisfiability of $\varphi$.
\end{itemize}

\paragraph{Conclusion.}
In this work we have presented an updated version of the original formal description of the \mcsat calculus.
In contrast to the original version we describe the calculus on arbitrary term assignments and do not distinguish between propositional and theory rules.
We presented all relevant concepts, showed the calculus' rules as well as provided examples.

%% file: acks.tex
\begin{acknowledgments}
The authors thank Ahmed Irfan, Stéphane Graham-Lengrand, and Jasper Nalbach for discussion and feedback.
We further thank the anonymous reviewers for exceptional feedback.
This research was funded in whole or in part by the ERC Consolidator Grant ARTIST 101002685; the the WWTF ICT22-007 grant ForSmart;
the FWF 10.55776/ESP666 grant, and by the SBA Research COMET Center SBA-K1 NGC managed by the FFG.
\end{acknowledgments}

\section*{Declaration on Generative AI}

The authors have not employed any Generative AI tools.

%% file: proofs.tex
We adapt the soundness and correctness proof \cite{mcsat} (Theorem 1) to suite our definitions.
While the proof is more detailed and structured differently, the general ideas remains the same.

\begin{lemma}\label{lemma:sinkStates}
	Given a set of clauses $\mathcal{C}$, in any derivation of the transition system starting from the initial state $\initstate$, any terminating derivation of the transition system ends in \ssat or \susat.
\end{lemma}
\begin{proof}
	To show that the derivation cannot be stuck in a state outside \ssat and \susat, we show that, given an arbitrary state $S \notin\{\ssat,\susat\}$, we can apply a rule to reach a different state.
	We distinguish two cases on $S$:
	\begin{enumerate}
		\item $S = \state{M}$, i.e.\ $S$ is not in conflict. If there exists a term $t$ such that $\feasible(M,t)\neq\emptyset$ we apply rule \nameref{rule:Decide} to reach state $S'$. If all terms are assigned and $\satisfiable(M,\mathcal{C})$ holds, then we use rule \nameref{rule:Sat} to reach state $S' = \ssat$.
		Otherwise, if there are only unassigned terms with $\feasible(M,t)=\emptyset$, then clearly $\infeasible(M,\mathcal{C})$ and we use rule \nameref{rule:Conflict} to reach a conflict state $S'$.
		In any case we have $S\neq S'$.
		
		\item $S = \state{M}\vdash C$, i.e.\ $S$ is in conflict with clause $C$. Then we can apply rules \nameref{rule:Resolve}, \nameref{rule:Consume}, and \nameref{rule:Drop} exhaustively until we either reach a state $S'$ such that either rule \nameref{rule:Backjump} can be used to reach a state $\state{M}$ not in conflict or rule \nameref{rule:Unsat} to reach state $\susat$. In both cases the new state is different from $S$.
	\end{enumerate}
	Since this clearly holds for $S = \initstate$, the desired property holds for any reachable state by induction.
\end{proof}

\begin{lemma}\label{lemma:clausesImplied}
	Given a set of clauses $\mathcal{C}$, in any derivation of the transition system starting from the initial state $\initstate$, all $\,\vdash C$ clauses are implied by $\mathcal{C}$.
\end{lemma}
\begin{proof}
	Let the implication of a clause $C$ by a set of clauses $\mathcal C$ be denoted as $\mathcal C \vDash C$.
	Given a state $S$ such that it is either of form $\state{M}$ or $\state{M}\vdash C$ where $\mathcal C \vDash C$.
	We show that applying an arbitrary rule $\mathcal R$ to reach state $S'$, the property holds for $S'$ as well.
	We split on $\mathcal R$:
	\begin{itemize}
		\item Rule \nameref{rule:Conflict}: Given $S = \state{M}$ to apply \nameref{rule:Conflict}, we have that $C = \explain(M,\bot)$. By the requirements of $\explain$, $\mathcal C \vDash C$ or a valid lemma, thus $S' = \state{M}\vdash C$ fulfills the property.
		
		\item Rule \nameref{rule:Resolve}: Let $S=\state{\trail{M,x\propsto \alpha}} \vdash C$ and $J = \just(x)$. Then $J$ was created using $\explain$ on a previous application of the rule \nameref{rule:Propagate}, thus $\mathcal C \vDash J$ and $J = (\lnot l_1 \lor \dots \lor \lnot l_\ell \lor x = t)$. 
		By assumption, we have $\mathcal C \vDash C$.
		We need to show that $\mathcal C \vDash C'$ for $C' = C[x/t] \cup (\lnot l_1\lor\dots\lor\lnot l_\ell)$.
		From $\mathcal C \vDash J$ we have that each model $I$ of $\mathcal C$ either $l_i^I = \false$ or $x^I = t^I$.
		In the first case $C'^I = \true$ as $\lnot l_i \in C'$, in the second case, since $x^I = t^I$ together with $C^I = \true$, $C[x/t]^I = \true$ and, thus, $C'^I = \true$.

		\item Any other rule does not generate a new clause $\,\vdash C$ clause.
	\end{itemize}
	Since the property holds for the initial state $\initstate$ trivially, the property holds for any reachable state in the derivation by induction.
\end{proof}

\begin{lemma}\label{lemma:conflictFalse}
	Given a set of clauses $\mathcal{C}$, in any derivation of the transition system starting from the initial state $\initstate$, we have that $\nu_{\bar M}(C) = \false$ for the $\,\vdash C$ clauses.
\end{lemma}
\begin{proof}
	We perform the same induction as for Lemma~\ref{lemma:clausesImplied} and split on the applied rule $\mathcal R$:
	\begin{itemize}
		\item Rule \nameref{rule:Conflict}: By construction of $C = \explain(M, \bot)$ in the resulting state $S' = \state{M} \vdash C$ we have that $\nu_{\bar M}(C) = \false$.
		
		\item Rule \nameref{rule:Resolve}: Let $S=\state{\trail{M,x\propsto \alpha}} \vdash C$ and $J = \just(x)$. By the induction hypothesis we have $\nu_{\bar M}(C) = \false$.
		Since $J$ was created using $\explain$ on a previous application of the rule \nameref{rule:Propagate}, $J = (\lnot l_1 \lor \dots \lor \lnot l_\ell \lor x = t)$ and $\nu_{\bar{M}}(l_i) = \false$ for $1\leq i \leq \ell$, thus $\nu_{\bar M}(x) = \nu_{\bar M}(t)$.
		Therefore, $\nu_{\bar M}(C) = \nu_{\bar M}(C[x/t])$ and, therefore, $\nu_{\bar M}(C') = \false$.
			
		\item Rule \nameref{rule:Consume}: As the propagated term $x$ is not contained in $C$, we have $\nu_{\bar M}(C) = \false$ trivially.
		
		\item Rule \nameref{rule:Drop}: The condition holds for the updated trail trivially as the desired property is a prerequisite of the rule.
		
		\item Any other rule does not generate a state with a $\,\vdash C$ clause.
	\end{itemize}
	Since the property holds for the initial state $\initstate$ trivially, the property holds for any reachable state in the derivation by induction.
\end{proof}

\begin{lemma}[Congruence]
\label{lem:congruence}
Given a set of clauses $\mathcal{C}$, in any derivation of the transition
system starting from the initial state $\langle \llbracket\rrbracket, \mathcal{C}\rangle$,
for every reachable state with trail $M$, every uninterpreted function symbol
$f$, and every pair of $f(t_1,\dots,t_n)$,
$f(t_1',\dots,t_n')$ with
$\nu_{\bar M}\big(f(t_1,\dots,t_n)\big),
\nu_{\bar M}\big(f(t_1',\dots,t_n')\big) \neq \udef$, we have
\[
\big(\forall i.\ \nu_{\bar M}(t_i) = \nu_{\bar M}(t_i')\big)
\;\Longrightarrow\;
\nu_{\bar M}\big(f(t_1,\dots,t_n)\big) = \nu_{\bar M}\big(f(t_1',\dots,t_n')\big).
\]
\end{lemma}

\begin{proof}
By induction on the derivation. 
The base case
$\langle \llbracket\rrbracket, \mathcal{C}\rangle$ holds automatically.
Only \textsc{Decide}, \textsc{Propagate},
and the decision case of \textsc{Backjump} assign a new value to a term; in
each case the assigned term is handled by its owning plugin, which we
require to respect congruence. All other
rules only remove trail elements or leave $M$ unchanged, so the property is
inherited from the induction hypothesis. 
\end{proof}

\begin{lemma}[Model Existence]
\label{lem:model-existence}
Given a set of clauses $\mathcal{C}$ and a term assignment $S$ such that
$\satisfiable(S, \mathcal{C})$ holds, $\mathcal{C}$ is satisfiable.
\end{lemma}

\begin{proof}
Define interpretation $I$ to agree with $\bar S$ on $\terms(\mathcal{C})$: for
interpreted symbols this is fixed by the evaluation function, and for
uninterpreted $f$, set $f^I(\nu_{\bar S}(t_1),\dots,\nu_{\bar S}(t_n)) =
\nu_{\bar S}(f(t_1,\dots,t_n))$ for each $f(t_1,\dots,t_n) \in
\terms(\mathcal{C})$, extending arbitrarily elsewhere. This is
well-defined by Lemma~\ref{lem:congruence}, and by $\consistent(\bar
S)$ we get $t^I = \nu_{\bar S}(t)$ for every $t \in \terms(\mathcal{C})$
by induction on term structure. Since $\satisfiable(S,
\mathcal{C})$ gives $\nu_{\bar S}(C) = \true$ for every $C \in \mathcal{C}$,
we obtain $C^I = \true$ for every $C \in \mathcal{C}$, so $I$ is a
model of $\mathcal{C}$.
\end{proof}

\begin{lemma}[Termination]\label{lemma:termination}
	Any derivation of the transition system starting from the initial state $\initstate$ terminates.
\end{lemma}
\newcommand{\bs}{\ensuremath{\longrightarrow_{bs}}\xspace}
\begin{proof}
	It is obvious that the search leaves the conflict mode after a finite number of rule applications towards the search mode or the state \susat. This is because each rule in conflict mode removes an element from the trail.
	We define a big-step transition relation \bs that covers a transition from a search state, applying a finite number of conflict rules to the next search state.
	
	Since $\explain$ returns new terms from a finite basis, we can assume a finite set of terms $T$ that contains all possible terms.
	In order to show progress of the search, we define a partial order $M_1 \prec M_2$ on trails.
	A trail contains two different kinds of elements: propagations ($x\propsto \alpha$) and decision ($x\decidesto\alpha$).
	We consider propagations to be heavier than any decision and define the weight of a decision by the size of the term.
	We define the size of a term $t$ by $|\terms(t)|$ and the maximum term size $\Omega = \max_{t\in T} |\terms(t)|$.
	We capture this in a weight function $w$ as follows:
	\begin{alignat*}{1}
		w(x\propsto \alpha) &= \Omega + 1\\
		w(x\decidesto\alpha) &= |\terms(x)|
	\end{alignat*}
	We define $M_1 \prec M_2$ using a lexicographical partial order based on the weights of the trail elements by
	\begin{alignat*}{2}
		\emptytrail   &\prec M             &\quad\textbf{if }&M \neq \emptytrail \\
		\trail{a,M_1} &\prec \trail{b,M_2} &\quad\textbf{if }&w(a) < w(b) \textbf{ or } (w(a) = w(b) \land M_1 \prec M_2)
	\end{alignat*}
	Clearly, $\emptytrail$ is the minimal element and for all trails $M$ and $N\neq \emptytrail$ we have $M\prec \trail{M,N}$ and since $|T|$ is finite, there is a maximal element.
	Any rule that adds a new element to the trail is creating a bigger trail with regard to $\prec$. This already covers most rules.
	Now let us consider the big-step \bs transition from a state $\sstate{M_1}{\mathcal{C}_1}$ to $\sstate{M_2}{\mathcal{C}_2}$.
	Since the only rule to exit the conflict state is \nameref{rule:Backjump}, any big step \bs will end in exactly one application of that rule.
	Note that no rule in conflict analysis adds an element to the trail, with the sole exception of \nameref{rule:Backjump} that replaces one trail element. Thus, $M_1$ and $M_2$ share a common prefix followed by the one trail element $(x\decidesto\alpha) \in M_1$ that was replaced by $(x'\decidesto\alpha')\in M_2$.
	There are two cases in the application of \nameref{rule:Backjump}. In the first case the decision $x\decidesto\alpha$ is replaced with a propagation $x'\propsto\alpha'$, thus $M_1\prec M_2$ as $w(x\decidesto\alpha) < w(x'\propsto\alpha')$.
	In the second case, since there are multiple literals of $C$ which are evaluated to $\udef$ by unassigning $x$, there is at least one such literal $x \neq x'$ such that $x\in\terms(x')$ and thus $w(x\decidesto \alpha) < w(x'\decidesto\alpha')$. Therefore, we have that $M_1\prec M_2$.
	
	Regarding clause learning and forgetting:
	Observe that any application of the rule \nameref{rule:Learn} is part of one big step \bs application.
	Any application of \nameref{rule:Learn-Justification} can only happen immediately after an application of \nameref{rule:Propagate}, thus any trail after the application of the two rules (as a big-step) will be bigger than before as the propagation has been added.
	To justify the \nameref{rule:Forget} rule we define another partial order $\sstate{M_1}{\mathcal{C}_1}\ll\sstate{M_2}{\mathcal{C}_2}$ by
	\[ M_1\prec M_2 \textrm{ or } (M_1 = M_2 \textrm{ and } |\mathcal{C}_1| > |\mathcal{C}_2|) \]
	With this definition we have that removing a clause generates a bigger state with regard to $\ll$. A maximal element $\sstate{M_\textrm{max}}{\mathcal C_0}$ of $\ll$ is a maximal element of $\prec$ and $C_0$, the initial set of clauses.
	Since all rules are producing bigger states and we cannot increase forever, the termination of the system follow.	
\end{proof}

\noindent
With the above lemmas, we can finally proof Theorem~\ref{thm:term}.

\begin{proof}[Proof of Theorem~\ref{thm:term}]
	Assume a set of clauses $\mathcal{C}$ and an explanation function $\explain$.
	Starting from the initial state $\initstate$, by Lemma~\ref{lemma:sinkStates} and Lemma~\ref{lemma:termination} we have that all derivations terminate in either $\ssat$ or $\susat$.
	In case the derivation terminates in \ssat, the last applied rule was
	\nameref{rule:Sat}, i.e., $\satisfiable(M,\mathcal{C})$ holds; by
	Lemma~\ref{lem:model-existence}, $\mathcal{C}$ is satisfiable.
	Thus we have that the original problem is indeed satisfied.
	
	Otherwise, if the derivation terminates in \susat, then rule \nameref{rule:Unsat} was applied, and the previous state was $\state{M_0}\vdash C$ where $M_0$ contains propagations only. By Lemma~\ref{lemma:conflictFalse}, $\nu_{\bar{M}}(C) = \false$.
	Resolving $C$ with all remaining propagations on $M_0$ to $C'$, we have $\nu_{\bar\emptyset}(C') = \false$. Then by Lemma~\ref{lemma:clausesImplied}, $\mathcal C$ implies $C$. Thus, there is no interpretation that satisfies $\mathcal C$ and the search problem is unsatisfiable.
\end{proof}

%% file: example.tex
This section provides further examples of different theories including theory combination using the conversion functions between different sorts.

We denote the sort of $m$-bit vectors as $\mathcal{S}_{BV}^{[m]}$, with the corresponding domain $\dom(\mathcal{S}_{BV}^{[m]}) = \{0,1\}^m,$ i.e. the set of all bit strings of length $m$.
Sorts for different sizes $m$ are disjoint, but they share a common set of function symbols:
\begin{enumerate*}[label=(\roman*)]
	\item Constants for each bit-vector of fixed width, e.g., \texttt{\#b1011} for $m=4$,
	\item Bitwise operators such as $\mathsf{and}$, $\bvor$, $\mathsf{xor}$, and $\mathsf{not}$,
	\item arithmetic operators, %
	\item as well as comparison operators, as usual, and
	\item conversion operators such as $\mathsf{bv2int}$, $\mathsf{bv2nat}$, and $\mathsf{nat2bv}$
\end{enumerate*}.

For a prime power $q = p^n$ with $p$ prime and $n \geq 1$, we denote the finite field with $q$ elements as $\mathbb{F}_q$.
For $n=1$, $\mathbb{F}_q$ corresponds to the integers modulo $p$; for $n > 1$, the field is constructed as a quotient ring of polynomials over $\mathbb{F}_p$ modulo an irreducible polynomial of degree $n$.
To indicate that a constant $x \in \mathbb{F}_p$ we write $x_p$.

\paragraph{Additional Example 1}
Assume we have the following instance $\mathcal{C}_\alpha = \{C_1, C_2, C_3, C_4\}$:
\begin{align*}
	C_1 &= (g(x)\teq x) \\
	C_2 &= (x - 50 \teq 3) \\
	C_3 &= (z >_u \texttt{\#b10} \lorS f(g(g(x))) \teq 16) \\
	C_4 &= (g(f(53)) \tneq g(16) \lorS y_7^3\teq 2_7)
\end{align*}
Here, $x \in \mathbb{R}$, $y$ from $\mathbb{F}_7$, and $z$ from $\mathcal{S}_{BV}^2$ are uninterpreted constants. 
Furthermore, $f: \mathbb{R} \rightarrow\mathbb{R}$ and $g: \mathbb{R} \rightarrow\mathbb{R}$ are uninterpreted functions.
In our search, we obtain the following search states by applying the rules of the calculus:
\begin{tabbing}
	\hspace*{1em} \= \hspace*{1.5em} \= \hspace*{1em} \=  \hspace*{1em} \=  \
	\hspace*{1em} \=  \hspace*{1em} \=  \hspace*{1em} \=  \hspace*{18em} \= \
	\hspace*{5em} \kill
	\> $1$. \>$\langle\llbracket \rrbracket,\mathcal{C}_\alpha\rangle$ \\
	\>\> $\downarrow$ Propagate with $\just(g(x)\teq x) = (g(x)\teq x)$\\
	\> $2$. \>$\langle\llbracket (g(x)\teq x)^* \rrbracket,\mathcal{C}_\alpha\rangle$\\
	\>\> $\downarrow$ Propagate with $\just(x - 50 \teq 3) = (x - 50 \teq 3)$\\
	\> $3$. \>$\langle\llbracket (g(x)\teq x)^*, (x - 50 \teq 3)^* \rrbracket,\mathcal{C}_\alpha\rangle$\\
	\>\> $\downarrow$ Propagate with $\just(x \propsto 53) = (x - 50 \teq 3)$\\
	\> $4$. \>$\langle\llbracket (g(x)\teq x)^*, (x - 50 \teq 3)^*, x \propsto 53 \rrbracket,\mathcal{C}_\alpha\rangle$\\
	\>\> $\downarrow$ Decide\\
	\> $5$. \>$\langle\llbracket (g(x)\teq x)^*, (x - 50 \teq 3)^*, x \propsto 53, z\decidesto_1 \texttt{\#b01} \rrbracket,\mathcal{C}_\alpha\rangle$\\
	\>\> We introduce the trail $M_1 = \llbracket (g(x)\teq x)^*, (x - 50 \teq 3)^*, x \propsto 53, z\decidesto_1 \texttt{\#b01} \rrbracket$.\\
	\>\> $\downarrow$ Propagate with $\just(f(g(g(x))) \propsto 16) = (z >_u \texttt{\#b10} \lorS f(g(g(x))) \teq 16)$\\
	\> $6$. \>$\langle\llbracket M_1, f(g(g(x))) \propsto 16 \rrbracket,\mathcal{C}_\alpha\rangle$\\
	\>\> $\downarrow$ Propagate with $\just(g(f(53))\tneq g(16) \propsto \false) =$\\
	\>\> $(\lnot(g(x) \teq x) \lorS \lnot(x \teq 53) \lorS \lnot(f(g(g(x)))\teq16) \lorS \lnot(g(f(53))\tneq g(16)))$\\
	\> $7$. \>$\langle\llbracket M_1, f(g(g(x))) \propsto 16, \lnot(g(f(53))\tneq g(16))^* \rrbracket,\mathcal{C}_\alpha\rangle$\\
	\>\> $\downarrow$ Propagate with $\just(y_7^3 \teq 2_7 \propsto \true) = (g(f(53)) \tneq g(16) \lorS y_7^3\teq 2_7)$\\
	\> $8$. \>$\langle\llbracket M_1, f(g(g(x))) \propsto 16, \lnot(g(f(53))\tneq g(16))^*, (y_7^3 \teq 2_7)^* \rrbracket,\mathcal{C}_\alpha\rangle$\\
	\>\> $\downarrow$ Conflict from $\mathcal{T}_{FF}$\\
	\> $9$. \>$\langle\llbracket M_1, f(g(g(x))) \propsto 16, \lnot(g(f(53))\tneq g(16))^*, (y_7^3 \teq 2_7)^* \rrbracket,\mathcal{C}_\alpha\rangle \vdash (y_7^3 \tneq 2_7)$\\
	\>\> $\downarrow$ Resolve with $\just(y_7^3 \teq 2_7 \propsto \true)$\\
	\> $10$. \>$\langle\llbracket M_1, f(g(g(x))) \propsto 16, \lnot(g(f(53))\tneq g(16))^* \rrbracket,\mathcal{C}_\alpha\rangle \vdash (g(f(53))\tneq g(16))$\\
	\>\> $\downarrow$ Resolve with $\just(g(f(53))\tneq g(16) \propsto \false)$\\
	\> $11$. \>$\langle\llbracket M_1, f(g(g(x))) \propsto 16 \rrbracket,\mathcal{C}_\alpha\rangle \vdash (\lnot(g(x) \teq x) \lorS \lnot(x \teq 53) \lorS \lnot(f(g(g(x)))\teq16))$\\
	\>\> $\downarrow$ Resolve with $\just(f(g(g(x))) \propsto 16)$\\
	\> $12$. \>$\langle\llbracket (g(x)\teq x)^*, (x - 50 \teq 3)^*, x \propsto 53, z\decidesto_1 \texttt{\#b01} \rrbracket,\mathcal{C}_\alpha\rangle \vdash $\\
	\>\>\>$(\lnot(g(x) \teq x) \lorS \lnot(x \teq 53) \lorS z >_u \texttt{\#b10})$\\
	\>\> Let $\mathcal{C}'_\alpha = \mathcal{C}_\alpha \cup \{(\lnot(g(x) \teq x) \lorS \lnot(x \teq 53) \lorS z >_u \texttt{\#b10})\}$.\\
	\>\> $\downarrow$ Backjump\\
	\> $13$. \>$\langle\llbracket (g(x)\teq x)^*, (x - 50 \teq 3)^*, x \propsto 53, (z >_u \texttt{\#b10})^* \rrbracket,\mathcal{C}'_\alpha\rangle$\\
	\>\> $\downarrow$ Propagate with $\just(z \propsto \texttt{\#b11}) = (\lnot(z >_u \texttt{\#b10}) \lorS z \teq \texttt{\#b11})$\\
	\> $14$. \>$\langle\llbracket (g(x)\teq x)^*, (x - 50 \teq 3)^*, x \propsto 53, (z >_u \texttt{\#b10})^*, z \propsto \texttt{\#b11} \rrbracket,\mathcal{C}'_\alpha\rangle$\\
	\>\> Let $N_0 = \llbracket (g(x)\teq x)^*, (x - 50 \teq 3)^*, x \propsto 53, (z >_u \texttt{\#b10})^*, z \propsto \texttt{\#b11} \rrbracket$.\\
	\>\> $\downarrow$ Propagate with $\just(g(53) \propsto 53) = (\lnot(x \teq 53) \lorS \lnot(g(x) \teq x) \lorS g(53) \teq 53)$\\
	\> $15$. \>$\langle\llbracket N_0, g(53) \propsto 53 \rrbracket,\mathcal{C}'_\alpha\rangle$\\
	\>\> $\downarrow$ Decide\\
	\> $16$. \>$\langle\llbracket N_0, g(53) \propsto 53, g(16) \decidesto_1 0 \rrbracket,\mathcal{C}'_\alpha\rangle$\\
	\>\> $\downarrow$ Decide\\
	\> $17$. \>$\langle\llbracket N_0, g(53) \propsto 53, g(16) \decidesto_1 0, f(53) \decidesto_2 0 \rrbracket,\mathcal{C}'_\alpha\rangle$\\
	\>\> $\downarrow$ Decide\\
	\> $18$. \>$\langle\llbracket N_0, g(53) \propsto 53, g(16) \decidesto_1 0, f(53) \decidesto_2 0, g(0) \decidesto_3 53 \rrbracket,\mathcal{C}'_\alpha\rangle$\\
	\>\> $\downarrow$ Decide\\
	\> $19$. \>$\langle\llbracket N_0, g(53) \propsto 53, g(16) \decidesto_1 0, f(53) \decidesto_2 0, g(0) \decidesto_3 53, y_7 \decidesto_4 0_7 \rrbracket,\mathcal{C}'_\alpha\rangle$\\
	\>\> $\downarrow$ Sat\\
	\> $20$. \>$\langle \ssat, \llbracket N_0, g(53) \propsto 53, g(16) \decidesto_1 0, f(53) \decidesto_2 0, g(0) \decidesto_3 53, y_7 \decidesto_4 0_7 \rrbracket\rangle$
\end{tabbing}
In this example, we can see that the decision $z\decidesto_1 \texttt{\#b01}$ applied in line $4$ caused the conflict. 
The conflict was detected in line $8$ after three propagations. 
The propagation in line $6$ from the theory of uninterpreted functions uses that $g(x) \teq x \landS x \teq 53 \landS f(g(g(x)))\teq16$ implies $g(f(53))\teq g(16)$, which is equivalent to $\lnot(g(f(53))\tneq g(16))$.
The subsequent resolution steps (lines $9-12$) systematically eliminate the assignments that led to the conflict, ultimately learning a new clause that forces a different value for $z$.
This learned constraint guides the search towards a valid model, which is successfully constructed through a series of decisions in the final steps (lines $16-19$).

\paragraph{Additional Example 2}
Assume we have the following instance $\mathcal{C}_\beta = \{C_1, C_2, C_3\}$:
\begin{align*}
	C_1 &= (z_{53} - z_{53}^2 \teq 51_{53}) \\
	C_2 &= (\mathsf{nat2bv}(\mathsf{bv2nat}(x))\teq y) \\
	C_3 &= (\lnot(f(x) \teq f(y)))
\end{align*}
Here, $z$ from $\mathbb{F}_{53}$ and $x, y$ from $\mathcal{S}_{BV}^8$ are uninterpreted constants. 
Furthermore, \mbox{$f: \mathcal{S}_{BV}^8 \rightarrow\mathbb{Z}$} is an uninterpreted function of arity $1$.
By applying the transition rules of our calculus, we obtain the following trace of search states:
\begin{tabbing}
	\hspace*{1em} \= \hspace*{1.5em} \= \hspace*{1em} \=  \hspace*{1em} \=  \
	\hspace*{1em} \=  \hspace*{1em} \=  \hspace*{1em} \=  \hspace*{18em} \= \
	\hspace*{5em} \kill
	\> $1$. \>$\langle\llbracket \rrbracket,\mathcal{C}_\beta\rangle$ \\
	\>\> $\downarrow$ Propagate with $\just(z_{53} - z_{53}^2 \teq 51_{53}) = (z_{53} - z_{53}^2 \teq 51_{53})$\\
	\> $2$. \>$\langle\llbracket (z_{53} - z_{53}^2 \teq 51_{53})^* \rrbracket,\mathcal{C}_\beta\rangle$\\
	\>\> $\downarrow$ Propagate with $\just(\mathsf{nat2bv}(\mathsf{bv2nat}(x))\teq y) = (\mathsf{nat2bv}(\mathsf{bv2nat}(x))\teq y)$\\
	\> $3$. \>$\langle\llbracket (z_{53} - z_{53}^2 \teq 51_{53})^*, (\mathsf{nat2bv}(\mathsf{bv2nat}(x))\teq y)^* \rrbracket,\mathcal{C}_\beta\rangle$\\
	\>\> $\downarrow$ Propagate with $\just(f(x) \teq f(y)) = (\lnot(f(x) \teq f(y)))$\\
	\> $4$. \>$\langle\llbracket (z_{53} - z_{53}^2 \teq 51_{53})^*, (\mathsf{nat2bv}(\mathsf{bv2nat}(x))\teq y)^*, \lnot(f(x) \teq f(y))^* \rrbracket,\mathcal{C}_\beta\rangle$\\
	\>\> $\downarrow$ Decide\\
	\> $5$. \>$\langle\llbracket (z_{53} - z_{53}^2 \teq 51_{53})^*, (\mathsf{nat2bv}(\mathsf{bv2nat}(x))\teq y)^*, \lnot(f(x) \teq f(y))^*, z_{53} \decidesto_1 2_{53} \rrbracket,\mathcal{C}_\beta\rangle$\\
	\>\> Let $M_1 = \llbracket (z_{53} - z_{53}^2 \teq 51_{53})^*, (\mathsf{nat2bv}(\mathsf{bv2nat}(x))\teq y)^*, \lnot(f(x) \teq f(y))^*, z_{53} \decidesto_1 2_{53} \rrbracket$.\\
	\>\> $\downarrow$ Decide\\
	\> $6$. \>$\langle\llbracket M_1,x\decidesto_2 \texttt{\#b01100010} \rrbracket,\mathcal{C}_\beta\rangle$ \\
	\>\> $\downarrow$ Propagate with $\just(y) = (\mathsf{nat2bv}(\mathsf{bv2nat}(x))\teq y)$\\
	\> $7$. \>$\langle\llbracket M_1,x\decidesto_2 \texttt{\#b01100010}, y \propsto \texttt{\#b01100010} \rrbracket,\mathcal{C}_\beta\rangle$ \\
	\>\> $\downarrow$ Conflict from $\mathcal{T}_{UF}$\\ 
	\> $8$. \>$\langle\llbracket M_1,x\decidesto_2 \texttt{\#b01100010}, y \propsto \texttt{\#b01100010} \rrbracket,\mathcal{C}_\beta\rangle \vdash (\lnot (x\teq y) )$ \\
	\>\> $\downarrow$ Resolve with $\just(y) = (\mathsf{nat2bv}(\mathsf{bv2nat}(x))\teq y)$\\
	\> $9$. \>$\langle\llbracket M_1,x\decidesto_2 \texttt{\#b01100010} \rrbracket,\mathcal{C}_\beta\rangle \vdash (\lnot (x\teq \mathsf{nat2bv}(\mathsf{bv2nat}(x))) )$ \\
	\>\> Let $\mathcal{C}'_\beta = \mathcal{C}_\beta \cup \{(\lnot (x\teq \mathsf{nat2bv}(\mathsf{bv2nat}(x))))\}$.\\
	\>\> $\downarrow$ Backjump\\
	\> $10$. \>$\langle\llbracket M_1,\lnot (x\teq \mathsf{nat2bv}(\mathsf{bv2nat}(x)))^* \rrbracket,\mathcal{C}'_\beta\rangle$ \\
	\>\> $\downarrow$ Conflict from $\mathcal{T}_{BV}$\\ 
	\> $11$. \>$\langle\llbracket M_1,\lnot(x\teq \mathsf{nat2bv}(\mathsf{bv2nat}(x)))^* \rrbracket,\mathcal{C}'_\beta\rangle \vdash (x\teq \mathsf{nat2bv}(\mathsf{bv2nat}(x)))$\\
	\>\> $\downarrow$ Resolve with $\just(x\teq \mathsf{nat2bv}(\mathsf{bv2nat}(x))) = (x\teq \mathsf{nat2bv}(\mathsf{bv2nat}(x)))$\\
	\> $12$. \>$\langle\llbracket (z_{53} - z_{53}^2 \teq 51_{53})^*, (\mathsf{nat2bv}(\mathsf{bv2nat}(x))\teq y)^*, \lnot(f(x) \teq f(y))^*, z_{53} \decidesto_1 2_{53} \rrbracket,\mathcal{C}'_\beta\rangle \vdash (\bot)$\\
	\>\> $\downarrow$ Drop\\
	\> $13$. \>$\langle\llbracket (z_{53} - z_{53}^2 \teq 51_{53})^*, (\mathsf{nat2bv}(\mathsf{bv2nat}(x))\teq y)^*, \lnot(f(x) \teq f(y))^*\rrbracket,\mathcal{C}'_\beta\rangle \vdash (\bot)$\\
	\>\> $\downarrow$ Unsat\\
	\> $14.$ \>$\susat$
\end{tabbing}
The instance $\mathcal{C}_\beta$ is $\susat$ because of the clauses $C_2 = (\mathsf{nat2bv}(\mathsf{bv2nat}(x))\teq y)$ and $C_3 = (\lnot(f(x) \teq f(y)))$.
The core contradiction arises because $C_2$ equates $y$ to a function of $x$ that is semantically equivalent to $x$ itself, effectively forcing $x = y$. This directly contradicts $C_3$, which requires $f(x) \neq f(y)$ and therefore $x \neq y$.
$C_2$ forces $x = y$ to hold, and $C_3$ implies $x \neq y$, which leads to the unsatisfiability.
In line $9$, the Backjump rule is applied to learn $(\lnot (x\teq \mathsf{nat2bv}(\mathsf{bv2nat}(x))))$ and propagate \mbox{$(x\teq \mathsf{nat2bv}(\mathsf{bv2nat}(x))) \propsto \false$} with the justification $\just(x\teq \mathsf{nat2bv}(\mathsf{bv2nat}(x))) = (\lnot (x\teq \mathsf{nat2bv}(\mathsf{bv2nat}(x))))$.
The learned clause $(\lnot (x\teq \mathsf{nat2bv}(\mathsf{bv2nat}(x))))$ captures the fundamental inconsistency that the term $\mathsf{nat2bv}(\mathsf{bv2nat}(x))$ cannot be equal to $x$ under the constraints imposed by the other clauses.
This propagated literal immediately leads to the conflict in line $10$ because there is no $\feasible$ value for $x$.

\paragraph{Additional Example 3}
Let $\mathcal{C}_\gamma$ be the following instance:
\begin{align*}
	C_1 &= (2x \leq 37) \\
	C_2 &= (\mathsf{bv2nat}(y) < x - 3) \\
	C_3 &= (\mathsf{nat2bv}(x)\teq \bvor(y,\texttt{\#b0011}) \lorS 20 - x < 2)
\end{align*}
Here, $x \in \mathbb{Z}$ and $y$ from $\mathcal{S}_{BV}^4$ are uninterpreted constants. 
By applying the rules of our calculus, we obtain the following trace of search states:
\begin{tabbing}
	\hspace*{1em} \= \hspace*{1.5em} \= \hspace*{1em} \=  \hspace*{1em} \=  \
	\hspace*{1em} \=  \hspace*{1em} \=  \hspace*{1em} \=  \hspace*{18em} \= \
	\hspace*{5em} \kill
	\> $1$. \>$\langle\llbracket \rrbracket,\mathcal{C}_\gamma\rangle$ \\
	\>\> $\downarrow$ Propagate with $\just(2x \leq 37) = (2x \leq 37)$\\
	\> $2$. \>$\langle\llbracket (2x \leq 37)^* \rrbracket,\mathcal{C}_\gamma\rangle$\\
	\>\> $\downarrow$ Propagate with $\just(\mathsf{bv2nat}(y) < x - 3) = (\mathsf{bv2nat}(y) < x - 3)$\\
	\> $3$. \>$\langle\llbracket  (2x \leq 37)^*, (\mathsf{bv2nat}(y) < x - 3)^* \rrbracket,\mathcal{C}_\gamma\rangle$\\
	\>\> $\downarrow$ Decide\\
	\> $4$. \>$\langle\llbracket (2x \leq 37)^*, (\mathsf{bv2nat}(y) < x - 3)^*, (20 - x < 2)_1 \rrbracket,\mathcal{C}_\gamma\rangle$\\
	\>\> $\downarrow$ Conflict\\
	\> $5$. \>$\langle\llbracket (2x \leq 37)^*, (\mathsf{bv2nat}(y) < x - 3)^*, (20 - x < 2)_1 \rrbracket,\mathcal{C}_\gamma\rangle \vdash $\\
	\>\>\>\>$(\lnot(2x \leq 37) \lorS \lnot(20 - x < 2))$\\
	\>\> Let $\mathcal{C}'_\gamma = \mathcal{C}_\gamma \cup \{(\lnot(2x \leq 37) \lorS \lnot(20 - x < 2))\}$.\\
	\>\> $\downarrow$ Backjump\\
	\> $6$. \>$\langle\llbracket (2x \leq 37)^*, (\mathsf{bv2nat}(y) < x - 3)^*, \lnot(20 - x < 2)^* \rrbracket,\mathcal{C}'_\gamma\rangle$\\
	\>\> Let $M = \llbracket (2x \leq 37)^*, (\mathsf{bv2nat}(y) < x - 3)^*, \lnot(20 - x < 2)^* \rrbracket$.\\
	\>\> $\downarrow$ Propagate with $\just(\mathsf{nat2bv}(x)\teq \bvor(y,\texttt{\#b0011})) =$\\
	\>\>\>\>$ (\mathsf{nat2bv}(x)\teq \bvor(y,\texttt{\#b0011}) \lorS 20 - x < 2)$\\
	\> $7$. \>$\langle\llbracket M, (\mathsf{nat2bv}(x)\teq \bvor(y,\texttt{\#b0011}))^* \rrbracket,\mathcal{C}'_\gamma\rangle$\\
	\>\> $\downarrow$ Decide\\
	\> $8$. \>$\langle\llbracket M, (\mathsf{nat2bv}(x)\teq \bvor(y,\texttt{\#b0011}))^*, x \decidesto_1 -1 \rrbracket,\mathcal{C}'_\gamma\rangle$\\
	\>\> $\downarrow$ Conflict\\
	\> $9$. \>$\langle\llbracket M, (\mathsf{nat2bv}(x)\teq \bvor(y,\texttt{\#b0011}))^*, x \decidesto_1 -1 \rrbracket,\mathcal{C}'_\gamma\rangle\vdash (0<x-3)$\\
	\>\> Let $\mathcal{C}''_\gamma = \mathcal{C}'_\gamma \cup \{(0<x-3)\}$.\\
	\>\> $\downarrow$ Backjump\\
	\> $10$. \>$\langle\llbracket M, (\mathsf{nat2bv}(x)\teq \bvor(y,\texttt{\#b0011}))^*, (0<x-3)^* \rrbracket,\mathcal{C}''_\gamma\rangle$\\
	\>\> $\downarrow$ Decide\\
	\> $11$. \>$\langle\llbracket M, (\mathsf{nat2bv}(x)\teq \bvor(y,\texttt{\#b0011}))^*, (0<x-3)^*, x \decidesto_1 7 \rrbracket,\mathcal{C}''_\gamma\rangle$\\
	\>\> $\downarrow$ Conflict\\
	\> $12$. \>$\langle\llbracket M, (\mathsf{nat2bv}(x)\teq \bvor(y,\texttt{\#b0011}))^*, (0<x-3)^*, x \decidesto_1 7 \rrbracket,\mathcal{C}''_\gamma\rangle \vdash $\\
	\>\>\>\>$ (\lnot(2x \leq 37)\lorS\lnot(\mathsf{bv2nat}(y) < x - 3) \lorS \lnot(\mathsf{nat2bv}(x)\teq \bvor(y,\texttt{\#b0011})))$\\
	\>\> $\downarrow$ Drop\\
	\> $13$. \>$\langle\llbracket M, (\mathsf{nat2bv}(x)\teq \bvor(y,\texttt{\#b0011}))^*, (0<x-3)^* \rrbracket,\mathcal{C}''_\gamma\rangle \vdash$\\
	\>\>\>\>$(\lnot(2x \leq 37)\lorS\lnot(\mathsf{bv2nat}(y) < x - 3) \lorS \lnot(\mathsf{nat2bv}(x)\teq \bvor(y,\texttt{\#b0011})))$\\
	\>\> $\downarrow$ Unsat\\
	\> $14.$ \>$\susat$
\end{tabbing}
We have shown that the instance $\mathcal{C}_\gamma$ is $\susat$.
In this example, the final conflict in line $11$ was only detected after a different conflict was managed.
The final conflict could have been detected in line $7$ already.
For the conflict detected in line $8$, we use the conflict clause $(0<x-3)$, which is generated by an oracle that uses the fact $\mathsf{bv2nat}(y) \geq 0$.
The clause $(0<x-3)$ learned from the first conflict effectively prunes the search space by restricting $x$ to values greater than $3$.
Here the solver explores multiple possible assignments for $x$ until the final conflict is derived.
The final conflict demonstrates that even when $x$ is assigned a value like $7$, which satisfies all individual theory constraints, the specific combination of constraints from $C_1$, $C_2$, and $C_3$ is inherently unsatisfiable.
The Drop rule can be applied in line $12$ because all literals of the conflict clause $(\lnot(2x \leq 37)\lorS\lnot(\mathsf{bv2nat}(y) < x - 3) \lorS \lnot(\mathsf{nat2bv}(x)\teq \bvor(y,\texttt{\#b0011})))$ evaluate to $\false$ under the current trail.
This is due to the fact that their negations $2x \leq 37, \mathsf{bv2nat}(y) < x - 3$ and $\mathsf{nat2bv}(x)\teq \bvor(y,\texttt{\#b0011})$ are on the current trail and evaluate to $\true$.